\documentclass[pra,twocolumn,superscriptaddress,nofootinbib]{revtex4-2}
\usepackage{array}[=2016-10-06]
\usepackage[letterpaper, left=0.8in, right=0.8in, top=0.9in, bottom=0.9in]{geometry}
\usepackage{titlesec}
\usepackage{booktabs, tikz, tabularx}

\usepackage{cmap} 
\usepackage[utf8]{inputenc}
\usepackage[english]{babel}
\usepackage[T1]{fontenc}
\usepackage{amsmath}
\usepackage{amsfonts}
\usepackage{bm}
\usepackage{soul}
\usepackage{enumitem}
\usepackage[dvipsnames]{xcolor}
\definecolor{blueviolet}{rgb}{0.2, 0.2, 0.6}
\definecolor{webgreen}{rgb}{0,.5,0}
\definecolor{webbrown}{rgb}{.6,0,0}
\usepackage[pdftex,
	bookmarks=false,
	colorlinks=true, 
	urlcolor=webbrown, 
	linkcolor=blueviolet, 
	citecolor=webgreen,
	pdfstartpage=1,
	pdfstartview={FitH},  
	bookmarksopen=false
	]{hyperref}
\usepackage{tikz}
\usepackage{natbib}
\usepackage{mathtools}
\usepackage{graphicx}
\usepackage{csquotes}
\usepackage{braket}
\usepackage{dsfont}
\usepackage{listings}
\allowdisplaybreaks

\usepackage{comment}

\usepackage{multirow}
\usepackage{amssymb}

\usepackage{amsthm}

\usepackage{color}
\usepackage{nicefrac}

\def\Lip{{\mathrm{Lip}}}
\newcommand{\msf}[1]{{\mathsf{#1}}}

\DeclareFixedFont{\ttb}{T1}{txtt}{bx}{n}{9} 
\DeclareFixedFont{\ttm}{T1}{txtt}{m}{n}{9}  

\usepackage{color}
\definecolor{deepblue}{rgb}{0,0,0.5}
\definecolor{deepred}{rgb}{0.6,0,0}
\definecolor{deepgreen}{rgb}{0,0.5,0}

\newcommand\pythonstyle{\lstset{
language=Python,
basicstyle=\ttm,
morekeywords={self},              
keywordstyle=\ttb\color{deepblue},
emph={MyClass,__init__},          
emphstyle=\ttb\color{deepred},    
stringstyle=\color{deepgreen},
frame=tb,                         
showstringspaces=false
}}

\lstnewenvironment{python}[1][]
{
\pythonstyle
\lstset{#1}
}
{}

\newcommand\pythoninline[1]{{\pythonstyle\lstinline!#1!}}

\def\bra#1{\ensuremath{\mathinner{\langle{#1}|}}}
\def\ket#1{\ensuremath{\mathinner{|{#1}\rangle}}}

\newcommand{\Husimi}{\mathrm{H}}

\newcommand{\mcal}[1]{\mathcal{#1}}

\newcommand{\bbR}{\mathbb{R}}

\newcommand{\diff}{\,\mathrm{d}}

\newcommand{\rmi}{\mathrm{i}}

\DeclareMathOperator{\Tr}{tr}

\theoremstyle{plain}
\newtheorem{theorem}{Theorem}
\newtheorem{corollary}{Corollary}
\newtheorem{lemma}{Lemma}

\newtheorem{proposition}{Proposition}
\newtheorem{definition}{Definition}

\theoremstyle{definition}

\newtheorem{example}{Example}

\theoremstyle{remark}
\newtheorem{remark}{Remark}

\newcommand{\Id}{I}

\usepackage[noend]{algpseudocode}
\usepackage{algorithm,algorithmicx}

\algrenewcommand\alglinenumber[1]{\sf\scriptsize\color{blue}{#1}}
\algrenewcommand\algorithmicrequire{\textbf{Input:}}
\algrenewcommand\algorithmicensure{\textbf{Output:}}

\def\eps{{\varepsilon}}
\def\One{{\mathds{1}}}

\newcommand{\Op}{{\operatorname*{Op}}}

\def\bra#1{\mathinner{\langle{#1}|}}
\def\ket#1{\mathinner{|{#1}\rangle}}

\usepackage[dvipsnames]{xcolor}

\makeatletter
\newif\ifappendixtoc
\let\orig@addcontentsline\addcontentsline

\renewcommand{\addcontentsline}[3]{%
  \orig@addcontentsline{#1}{#2}{#3}%
  \ifappendixtoc
    \def\atoc@a{#1}%
    \def\atoc@b{toc}%
    \ifx\atoc@a\atoc@b
      \orig@addcontentsline{atoc}{#2}{#3}%
    \fi
  \fi
}

\newcommand{\appendixtableofcontents}{%
  \section*{Appendix Contents}%
  \@starttoc{atoc}%
}
\makeatother

\begin{document}

\title{Quantum Chaos and Quantum Optimal Transport}

\author{Jordan Cotler}
    \email{jcotler@fas.harvard.edu}
    \affiliation{Department of Physics, Harvard University, Cambridge, MA 02138 USA}
\author{Felipe Hern\'{a}ndez}
    \email{felipeh@psu.edu}
    \affiliation{Department of Mathematics, Massachusetts Institute of Technology, Cambridge, MA 02139 USA}
    \affiliation{Department of Mathematics, Penn State University, University Park, PA 16802 USA}
\date{\today}

\begin{abstract}
Chaos in classical systems can be characterized by Lyapunov exponents that measure the exponential divergence of nearby trajectories, but directly extending this framework to quantum mechanics has been a persistent challenge. The wavelike nature of quantum states and the non-commutative geometry of quantum phase space obstruct a straightforward generalization of classical chaos theory. Here we develop a rigorous approach to quantum chaos by leveraging quantum optimal transport theory, which provides the missing geometric foundation for measuring distances between extended quantum distributions. We define quantum Lyapunov exponents that naturally avoid divergences encountered when na\"{i}vely generalizing classical exponents, and show that in the semiclassical limit they recover the classical global expansion rate, and hence the usual maximal Lyapunov exponent when they coincide. Our framework provides a tight connection between the divergence of classical trajectories and semiclassical phase space evolution, and additionally clarifies the role of out-of-time-order correlators as diagnostics of quantum chaos. These results establish quantum optimal transport as a unifying mathematical foundation for quantum chaos theory, providing new tools to characterize dynamical behavior across the full range of quantum dynamics from simple few-body models to complex many-body systems.
\end{abstract}

\maketitle

\section{Introduction}

Natural systems typically exhibit chaotic dynamics, characterized by extreme sensitivity to initial conditions. Such behavior is pervasive across distance scales, from the three-body problem in celestial mechanics to the motion of electrons in disordered metals~\cite{haake1991quantum, cvitanovic2005chaos}. Classical chaos theory has matured into a well-established field, built on rigorous mathematical foundations including Lyapunov exponents, ergodic theory, and symbolic dynamics~\cite{cvitanovic2005chaos, robinson1998dynamical}. However, despite considerable effort, a comparable foundational understanding of chaos in quantum mechanical systems remains elusive. Current approaches to quantum chaos, sometimes termed ``quantum chaology''~\cite{berry1989quantum}, comprise a collection of phenomena and diagnostic tools, including random matrix theory~\cite{beenakker1997random, guhr1998random}, semiclassical expansions~\cite{gutzwiller2013chaos, berry1972semiclassical, berry1985semiclassical, haake1991quantum}, and large-$N$ methods~\cite{yaffe1982large}, rather than a unified conceptual framework. Given that our universe is quantum mechanical, developing a systematic theory of quantum chaos is scientifically pressing.

\begin{figure}[t!]
    \centering
    \includegraphics[width=.9\linewidth]{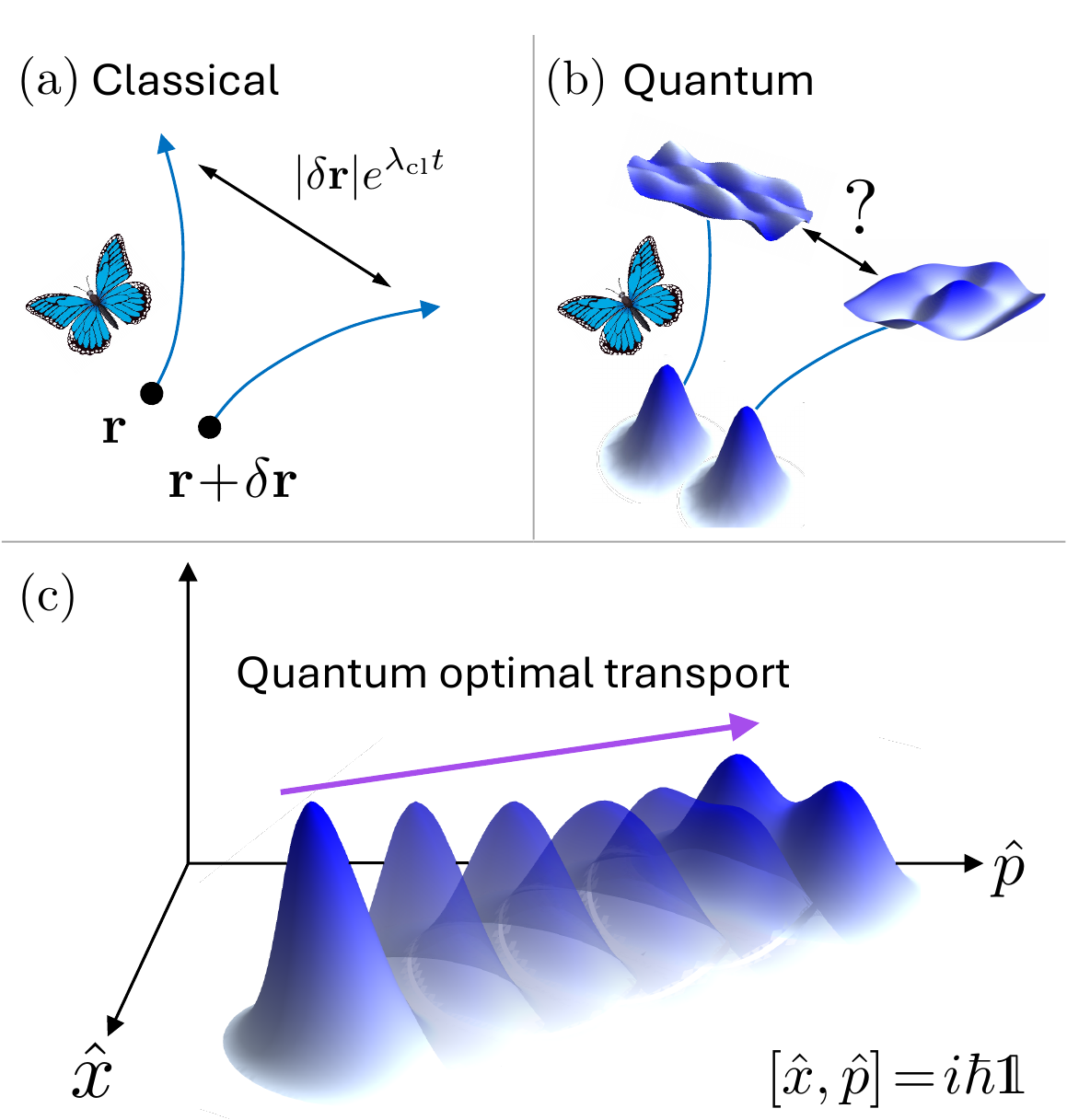}
    \caption{(a) In classical chaos, trajectories from initial conditions $\textbf{r}$ and $\textbf{r} + \delta \textbf{r}$ diverge as $|\delta \textbf{r}|\,e^{\lambda_{\rm cl} t}$, where $\lambda_{\rm cl}$ is the maximal Lyapunov exponent. (b) Quantum states are extended in phase space, making it unclear how to measure distances between evolved states. (c) Quantum optimal transport provides a natural distance metric by quantifying the cost of transporting between quantum configurations while respecting noncommutativity of quantum variables.}
    \label{fig:conceptual}
\end{figure}

In this paper, we define a notion of Lyapunov exponent for quantum systems that closely parallels the classical definition in terms of exponential divergence of trajectories. Our approach addresses a long-standing conceptual challenge in formulating quantum chaos: the absence of an appropriate distance measure on quantum states that incorporates both the geometry of Hilbert space and the geometry of quantum phase space. To achieve this, we apply tools from semiclassical analysis~\cite{zworski2012semiclassical} to the recently developed theory of quantum optimal transport~\cite{beatty2025wasserstein}, which provides precisely such a distance measure. Our work reveals subtle features of quantum phase space that make defining natural quantum Lyapunov exponents challenging. Indeed, a direct quantum generalization of the ordinary Lyapunov exponent can be infinite for reasons unrelated to dynamical instability. We show that our refined quantum Lyapunov exponent is finite and recovers the classical global expansion rate in the semiclassical limit, and therefore the usual maximal Lyapunov exponent in common settings where the two agree, for example when the relevant instability is confined to a compact region.

The essential difficulty in defining quantum chaos can be illustrated through a direct comparison with classical chaos. Classical chaotic dynamics are typically visualized in \textit{phase space}, where each point $\mathbf{r} = (x_1, \dots, x_D, p_1, \dots, p_D)$ represents the positions and momenta of particles. When a classical system evolves from slightly different initial conditions, separated by a small perturbation $\delta \mathbf{r}$, the resulting trajectories diverge exponentially, with separation growing as $|\delta \mathbf{r}| e^{\lambda_{\rm cl} t}$, where $\lambda_{\rm cl}$ is the (classical) maximal \textit{Lyapunov exponent} (Fig.~\ref{fig:conceptual}(a)).

Quantum mechanics introduces complications to this picture. The canonical commutation relation $[\hat{x}, \hat{p}] = i \hbar \mathds{1}$ renders quantum phase space inherently non-commutative, intertwining position and momentum. The Heisenberg uncertainty principle $\Delta x \Delta p \geq \hbar / 2$ prevents quantum states from being localized to precise phase space points. Instead, quantum states are represented by extended distributions, namely Wigner functions, that evolve into intricate patterns under chaotic dynamics~\cite{curtright2012quantum} (Fig.~\ref{fig:conceptual}(b)). The fact that quantum states exist as extended distributions leads to the question: how do we measure the `distance' between quantum states to quantify chaos? Unlike classical point particles whose trajectories clearly separate over time, the overlap and interference of quantum distributions defy straightforward geometric measures of distance.

These considerations reveal two differences between quantum and classical dynamics: quantum phase space forms a \textit{non-commutative geometry}~\cite{curtright2012quantum, connes2010noncommutative}, and quantum wavefunctions are inherently \textit{extended distributions} rather than points. Defining a meaningful distance between quantum states therefore requires synthesizing the geometric structures of both classical phase space and Hilbert space in a quantum-mechanically consistent framework. Classical optimal transport theory is naturally suited to this problem, as it provides a general framework for measuring distances between extended distributions. The \textit{optimal transport distance}, also known as the ``earthmover distance'' or ``Wasserstein distance'', elegantly quantifies distances between extended distributions by imagining distributions as piles of soil and measuring the minimum work required to rearrange one pile into another~\cite{villani2008optimal}. This approach has significantly impacted the theory of partial differential equations~\cite{jordan1998variational,otto2001geometry} and contemporary machine learning~\cite{cuturi2013sinkhorn,arjovsky2017wasserstein}, providing a natural bridge between physical metric geometry and probability distributions. Crucially, optimal transport distance reduces to conventional geometric distances when distributions shrink to points.

To study an appropriate distance between quantum states, we leverage recent advances in \textit{quantum optimal transport theory}~\cite{carlen2014analog, golse2016mean, carlen2020non, de2021quantum, cole2023quantum, lafleche2023quantum, beatty2025wasserstein}, which generalizes classical optimal transport to incorporate the non-commutative geometry of quantum phase space and the extended nature of quantum states (Fig.~\ref{fig:conceptual}(c)). Related earlier constructions defined optimal transport distances between quantum states through classical optimal transport of their Husimi phase-space distributions~\cite{zyczkowski1998monge,zyczkowski2001monge}. Here we instead build on the modern noncommutative quantum optimal transport framework, whose intrinsic geometry of quantum states provides the structure needed to define quantum Lyapunov exponents in direct analogy with their classical counterparts.

We address several technical challenges while ensuring that our main results apply to any quantum optimal transport distance satisfying natural axioms. First, we prove that the quantum optimal transport distance agrees with the classical optimal transport distance up to additive $O(\hbar^{1/2})$ corrections in the semiclassical regime. For the Carlen--Maas distances~\cite{carlen2020non}, we additionally show that they obey our axioms. Second, we construct a quantum Lyapunov exponent that avoids the divergences encountered in a na\"{i}ve infinitesimal definition. Third, we show that this exponent equals the classical global expansion coefficient in the semiclassical limit, and hence the usual maximal Lyapunov exponent in the common settings where the two agree. Finally, we compare the quantum optimal transport exponent with out-of-time-order correlators (OTOCs)~\cite{larkin1969quasiclassical, shenker2014black, maldacena2016bound, cotler2018out}. A strength of our approach is its axiomatic foundation: the main result holds for any quantum optimal transport distance satisfying a short list of natural properties.

These results represent a significant step toward a unified theory of quantum chaos. By providing rigorous mathematical foundations that seamlessly connect quantum and classical perspectives, our framework opens new possibilities for characterizing dynamical behavior in diverse quantum systems from few-body models to complex many-body phenomena, and establishes quantum optimal transport as an essential tool for understanding the nature of chaos in our quantum universe.

\section{Results}

In Section~\ref{subsec:define} we construct a well-defined notion of quantum Lyapunov exponents using the quantum optimal transport framework, carefully addressing technical subtleties that cause na\"{i}ve approaches to diverge. We also demonstrate our approach in an explicit example for the Carlen--Maas distances~\cite{carlen2020non}, computing the quantum Lyapunov exponent exactly for the inverted quantum harmonic oscillator. In Section~\ref{subsec:synthesis} we prove that the quantum Lyapunov exponent equals the classical global expansion coefficient in the semiclassical limit and establish a precise comparison with out-of-time-order correlators, showing that OTOC growth is bounded by the same classical instability and saturates this bound under an additional condition.

\subsection{Quantum optimal transport \\ \qquad and its classical limit}\label{subsec:classical}

We begin by establishing the mathematical framework for quantum optimal transport distances in phase-space dynamics. Let $\mathcal{D}$ denote the space of density matrices on $L^2(\mathbb{R}^D)$, and let $\hat{r} \equiv (\hat{x}_1,\ldots,\hat{x}_D,\hat{p}_1,\ldots,\hat{p}_D)$ denote the vector of canonical phase-space operators. For $\alpha \in \mathbb{R}^{2D}$, we define the phase-space translation $\tau_\alpha \equiv \exp(i\alpha^\intercal J\hat{r}/\hbar)$ and write $\mathcal{T}_\alpha[\rho] \equiv \tau_\alpha\rho\tau_\alpha^\dagger$. Moreover, let $|\alpha\rangle$ denote the coherent state centered at $\alpha$. Since coherent states resolve the identity, every quantum state $\rho$ defines a normalized and nonnegative phase-space distribution, called its Husimi distribution, given by $\Husimi_\rho(\alpha) \equiv (2\pi\hbar)^{-D}\langle\alpha|\rho|\alpha\rangle$. We denote the classical $p$-Wasserstein distance between the Husimi distributions of $\rho$ and $\sigma$ by $W_p(\Husimi_\rho,\Husimi_\sigma)$.

Quantum optimal transport distances are naturally defined on weighted spaces of density matrices. We denote by $\mcal{D}_p$ the states with finite $p$-th phase-space moment, namely those satisfying $\int_{\mathbb R^{2D}} |\alpha|^p \Husimi_\rho(\alpha)\diff\alpha < \infty$. We define a quantum distance $d_p : \mathcal{D}_p \times \mathcal{D}_p \to \mathbb{R}_+$ for $p \geq 1$ to be an \textit{admissible} generalization of the $p$-Wasserstein distance if it satisfies four natural axioms beyond those of a metric:
\begin{enumerate}[label=A\arabic*]
\item\textbf{Translation invariance}:
\label{ax:tr-invar}
\vspace{-.25cm}  \[d_p(\mathcal{T}_{\alpha}(\rho),\mathcal{T}_{\alpha}(\sigma)) = d_p(\rho,\sigma).\]
\item\textbf{Translation distance}: 
\label{ax:tr-id}
$d_p(\rho,\mathcal{T}_\alpha(\rho)) = |\alpha|$.
\item\label{ax:conv}\textbf{Double convexity}: If $\rho = \int \rho_\theta\,\mu(\theta)\diff\theta$ and $\sigma = \int \sigma_\phi\,\nu(\phi)\diff\phi$, where $\mu$ and $\nu$ are probability distributions and $\Gamma(\theta,\phi)$ marginalizes to $\mu(\theta)$ and $\nu(\phi)$, then
\begin{align}
d_p(\rho,\sigma)^p \leq \int d_p(\rho_\theta,\sigma_\phi)^p \diff\Gamma(\theta,\phi)\,.
\end{align}
\vspace{-.5cm} 
\item\label{ax:dp}\textbf{Husimi data-processing bound}: \begin{align}
W_p(\Husimi_\rho,\! \Husimi_\sigma) \leq d_p(\rho,\sigma)\,.
\end{align}
\end{enumerate}
\noindent These axioms encode basic physical requirements: (i) the distance should be invariant under phase space translations; (ii) the distance between a state and its translated version should equal the translation distance; (iii) classical mixing of quantum states should not artificially inflate distances; and (iv) the Husimi map sends quantum states to classical phase-space distributions, and the classical $p$-Wasserstein distance between the resulting distributions should not exceed the quantum optimal transport distance between the original states.\footnote{Axiom~(iv) could be weakened to $W_p(\Husimi_\rho,\Husimi_\sigma)\leq C\,d_p(\rho,\sigma)$ for any fixed $C$ without affecting our results, since such a constant disappears upon taking the exponential growth rate. The examples considered here satisfy the bound with $C=1$.} Examples of admissible distances include the Husimi--Monge distance of~\cite{zyczkowski1998monge}, corresponding to $p = 1$ in our notation, and the infinite-dimensional versions of the Carlen--Maas 1-Wasserstein and 2-Wasserstein distances for $p = 1$ and $2$, respectively~\cite{carlen2020non}. While the original Carlen--Maas distances were developed for finite-dimensional Hilbert spaces, in the Appendix we develop the appropriate generalization to $L^2(\mathbb{R}^D)$.

We would now like to show that admissible quantum distances are also bounded above by the classical optimal transport distance in the semiclassical limit. To understand this limit, note that coherent states become increasingly localized as $\hbar \to 0$, with their phase space width scaling as $\hbar^{1/2}$, while the Husimi distributions of semiclassical families of quantum states approach classical probability distributions on phase space (although they are in fact valid probability distributions for any $\hbar$). It is therefore natural to expect that a quantum optimal transport distance $d_p(\rho,\sigma)$ should be bounded above, up to corrections at the coherent-state scale, by the classical $p$-Wasserstein distance between the corresponding Husimi distributions. We prove this is indeed the case:
\begin{theorem}\label{thm:upperbd1}
For any admissible distance $d_p$ we have
\begin{align}
 W_p(\Husimi_\rho,\Husimi_\sigma) \leq d_p(\rho,\sigma) \leq W_p(\Husimi_\rho,\Husimi_\sigma) + O(\hbar^{\frac{1}{2}}),
\end{align}
where $W_p$ is the (classical) $p$-Wasserstein distance for any $p \geq 1$.
\end{theorem}
\noindent This result establishes that admissible quantum optimal transport distances are well approximated by their classical counterparts, preserving the underlying geometry of phase space. While the lower bound comes from axiom (iv), the upper bound is a less obvious consequence of the other three axioms and is proved in the Appendix.

For any concrete quantum optimal transport distance, admissibility must be verified from its specific structure. In the Appendix, we establish admissibility for the Carlen--Maas distances $d_{\text{CM},1}$ and $d_{\text{CM},2}$, which also satisfy the following bounds:
\begin{lemma}
\label{lemma:CMsymp1}
Let $S \in \operatorname{Sp}(2D,\mathbb R)$ be a linear symplectic transformation, let $U_S$ be the corresponding metaplectic unitary, and define $\rho_S \equiv U_S \rho U_S^\dagger$ and $\sigma_S \equiv U_S \sigma U_S^\dagger$. Then we have
\begin{align*}
\begin{aligned}
\|S\|^{-1} d_{\text{\rm CM},p}(\rho,\sigma) &\leq d_{\text{\rm CM},p}(\rho_S,\sigma_S) \leq \|S\|\,d_{\text{\rm CM},p}(\rho,\sigma)
\end{aligned}
\end{align*}
for $p = 1,2$.
\end{lemma}
\noindent Thus linear canonical transformations stretch or compress the Carlen--Maas distances by at most the operator norm factor of the corresponding classical symplectic transformation.

\subsection{Defining quantum Lyapunov exponents}\label{subsec:define}

Having established our framework for quantum optimal transport distances, we now turn to defining quantum Lyapunov exponents. Our goal is to quantify how small perturbations in initial quantum states grow exponentially over time, mirroring the classical notion of trajectory divergence. For a quantum state $\rho$ evolving to $\rho(t)$ under unitary evolution or a more general quantum channel (with similar notation for $\sigma(t)$), a natural starting point would be to define the quantum Lyapunov exponent in direct analogy to the classical case\footnote{The direct classical counterpart of this finite-separation definition is the global expansion coefficient defined below. In many standard settings, this coincides with the usual maximal Lyapunov exponent.}:
\[
\lambda_{\text{q}} \overset{?}{=} \lim_{T\to\infty} \lim_{\hbar\to 0} \sup_{\rho,\sigma} \frac{1}{T} \log\!\left(\frac{d_p(\rho(T), \sigma(T))}{d_p(\rho(0),\sigma(0))}\right)
\]
However, this na\"{i}ve definition can be infinite for reasons unrelated to chaos. To see the issue, consider the Carlen--Maas 2-Wasserstein distance $d_{{\rm CM},2}$~\cite{carlen2020non}. In finite dimension, its geometry is smooth near full-rank states, but it becomes singular near states with zero eigenvalues; related singular behavior can occur in infinite dimension. Moreover, a family of nearby states $\rho_\varepsilon = \rho + \varepsilon\nu + O(\varepsilon^2)$ may initially satisfy $d_{{\rm CM},2}(\rho_\varepsilon,\rho) = O(\varepsilon)$, but a general Hamiltonian evolution need not preserve the perturbation directions with this linear behavior. A simple classical analogue already exhibits this singular behavior: near the boundary of the probability simplex, an $O(\varepsilon)$ perturbation of the probability distribution can have transport distance of order $\varepsilon\sqrt{\log(1/\varepsilon)}$, rather than $O(\varepsilon)$. Thus, if evolution carries an initially regular perturbation into such a singular direction, the ratio $d_{{\rm CM},2}(\rho_\varepsilon(T),\rho(T))/d_{{\rm CM},2}(\rho_\varepsilon,\rho)$ can grow as $\sqrt{\log(1/\varepsilon)}$ and diverge as $\varepsilon \to 0$. The supremum in the na\"{i}ve definition can therefore be infinite already at fixed time, even when the dynamics has no chaotic instability. This divergence reflects the microscopic geometry of the transport distance rather than genuine exponential sensitivity to initial conditions. We discuss this phenomenon further in Appendix~\ref{sec:finite-action-tangents}.

The singular behavior above occurs at separations comparable to or smaller than the intrinsic quantum resolution scale. We therefore restrict the Lyapunov problem to $\hbar$-dependent pairs of states whose initial transport distance is parametrically larger than $\hbar^{1/2}$. This scale has a simple phase-space interpretation: because $[\hat{x}_j,\hat{p}_k] = i\hbar\,\delta_{jk}$, coherent states resolve phase space only to a characteristic linear scale of order $\hbar^{1/2}$. Writing $d_{p,\hbar}$ to make the $\hbar$-dependence of the quantum optimal transport distance explicit, we use $d_{p,\hbar}(\rho_\hbar,\sigma_\hbar) = \omega(\hbar^{1/2})$ to mean that $(\rho_\hbar,\sigma_\hbar)_{\hbar>0}$ is a family satisfying $\lim_{\hbar \to 0}\hbar^{-1/2} d_{p,\hbar}(\rho_\hbar,\sigma_\hbar)=\infty$. At these separations, the $O(\hbar^{1/2})$ discrepancy between quantum and classical transport geometry is negligible relative to the initial distance, so the Lyapunov quotient does not probe the singular microscopic geometry described above.  These considerations motivate the following definition of the quantum Lyapunov exponent.

\begin{definition}[Quantum Lyapunov exponent]
\label{def:qlyap1}
For any admissible quantum optimal transport distance $d_{p,\hbar}$, we define the maximal quantum Lyapunov exponent by
\begin{widetext}
\begin{equation}
\label{eq:lambdaq-def}
\lambda_{\rm q} \equiv \lim_{T\to\infty}\sup_{d_{p,\hbar}(\rho_\hbar,\sigma_\hbar) = \omega(\hbar^{1/2})}\limsup_{\hbar\to 0}\frac{1}{T}\log\!\left(\frac{d_{p,\hbar}(\rho_\hbar(T),\sigma_\hbar(T))}{d_{p,\hbar}(\rho_\hbar(0),\sigma_\hbar(0))}\right).
\end{equation}
\end{widetext}
where the supremum is over $\hbar$-dependent families $(\rho_\hbar,\sigma_\hbar)_{\hbar>0}$ satisfying the asymptotic condition above. We suppress the explicit $\hbar$ subscript on $d_{p,\hbar}$ below.
\end{definition}
\noindent The natural classical counterpart of Eq.~\eqref{eq:lambdaq-def} is the global expansion coefficient
\[\lambda_{\rm exp} \equiv \lim_{T\to\infty}\sup_{x\neq y} \frac{1}{T}\log\!\left(\frac{|\Phi_T(x)-\Phi_T(y)|}{|x-y|}\right),\]
where $\Phi_T$ denotes the classical Hamiltonian flow for time $T$. This choice is dictated by the quantum definition: $d_{p}$ measures finite state separations, and the supremum over initial pairs is taken at fixed $T$ before the long-time limit. Equivalently, if $J(T) \equiv \sup_{x\neq y}|\Phi_T(x)-\Phi_T(y)|/|x-y|$ is the Lipschitz constant of the classical flow, then $\lambda_{\rm exp} = \lim_{T\to\infty} T^{-1}\log J(T)$, and $J(T)$ is precisely the factor controlling classical transport through $W_p((\Phi_T)_\#\mu,(\Phi_T)_\#\nu) \leq J(T)W_p(\mu,\nu)$. In general, $\lambda_{\rm exp}$ can differ from the usual maximal Lyapunov exponent $\lambda_{\rm cl} \equiv \sup_x \lim_{T\to\infty} T^{-1}\log\|D\Phi_T(x)\|$, because the initial pair maximizing finite-time expansion may depend on $T$. Thus $\lambda_{\rm exp}$ is the classical quantity naturally selected by the transport problem. In many physically natural settings, such as compact phase spaces or dynamics whose relevant instability is confined to a compact region, $\lambda_{\rm exp} = \lambda_{\rm cl}$. We will show below that $\lambda_{\rm q} = \lambda_{\rm exp}$, and hence that $\lambda_{\rm q} = \lambda_{\rm cl}$ in these settings.  In Appendix~\ref{sec:classical-ot-exponent}, we formulate a purely classical optimal transport analogue of Definition~\ref{def:qlyap1} and show that it likewise recovers $\lambda_{\rm exp}$.

Before proving general properties of the quantum Lyapunov exponent, we consider a concrete example. We examine the quantum inverted harmonic oscillator, a paradigmatic system exhibiting an exponential instability. The classical inverted oscillator, governed by $H(x,p) = \frac{1}{2}\,p^2 - \frac{1}{2} x^2$, features an unstable equilibrium at the origin with trajectories diverging exponentially at rate $\lambda_{\text{cl}} = 1$ (Fig.~\ref{fig:conceptual2}(b)). For the quantum system, we establish a precise result:
\begin{theorem}[Quantum Lyapunov exponent for the inverted harmonic oscillator]
\label{thm:inverted1}
Let $\rho(t) = e^{-i\hat H t/\hbar}\rho e^{i\hat H t/\hbar}$, and similarly for $\sigma(t)$, where $\hat H$ is the quantum inverted-harmonic-oscillator Hamiltonian. Then, for any admissible quantum optimal transport distance,
\[
\lambda_{\rm q} = \lambda_{\rm exp} = \lambda_{\rm cl} = 1.
\]
For the Carlen--Maas distances $d_{\rm CM,1}$ and $d_{\rm CM,2}$, the stronger finite-time bound
\[
d_{\text{\rm CM},p}(\rho(t),\sigma(t)) \leq e^t d_{\text{\rm CM},p}(\rho,\sigma)
\]
holds for $p = 1,2$ and is saturated by suitable translated pairs of states.
\end{theorem}
\noindent The equality $\lambda_{\rm q} = 1$ for a general admissible distance follows from the main result below, since the classical inverted oscillator has global expansion factor $J(t) = e^t$. For the Carlen--Maas distances, the finite-time estimate follows directly from Lemma~\ref{lemma:CMsymp1}. The classical flow is the symplectic transformation $S_t:(x,p)\mapsto(\cosh(t)x+\sinh(t)p,\sinh(t)x+\cosh(t)p)$, whose operator norm is $\|S_t\| = e^t$. To see that the bound is sharp, take the unstable unit vector $a = 2^{-1/2}(1,1)$, for which $S_t a = e^t a$. If $\sigma = \mcal T_a[\rho]$, then $\sigma(t) = \mcal T_{e^t a}[\rho(t)]$, and the translation-distance property gives $d_{\rm CM,p}(\rho(t),\sigma(t)) = e^t$.

While this inverted harmonic oscillator example provides valuable intuition, the quadratic Hamiltonian avoids the subtle sub-$\hbar$ structure of generic quantum phase space. In the following, we extend our analysis beyond quadratic Hamiltonians to a broad class of quantum dynamics.

\subsection{Synthesis of quantum and classical \\ \qquad Lyapunov exponents}\label{subsec:synthesis}

We now relate the quantum Lyapunov exponent in Definition~\ref{def:qlyap1} to the classical expansion coefficient $\lambda_{\rm exp}$ and to the exponent extracted from out-of-time-order correlators. Our main result shows that the quantum optimal transport exponent exactly recovers $\lambda_{\rm exp}$ in the semiclassical limit.

\begin{figure}[t!]
    \centering
    \includegraphics[width=\linewidth]{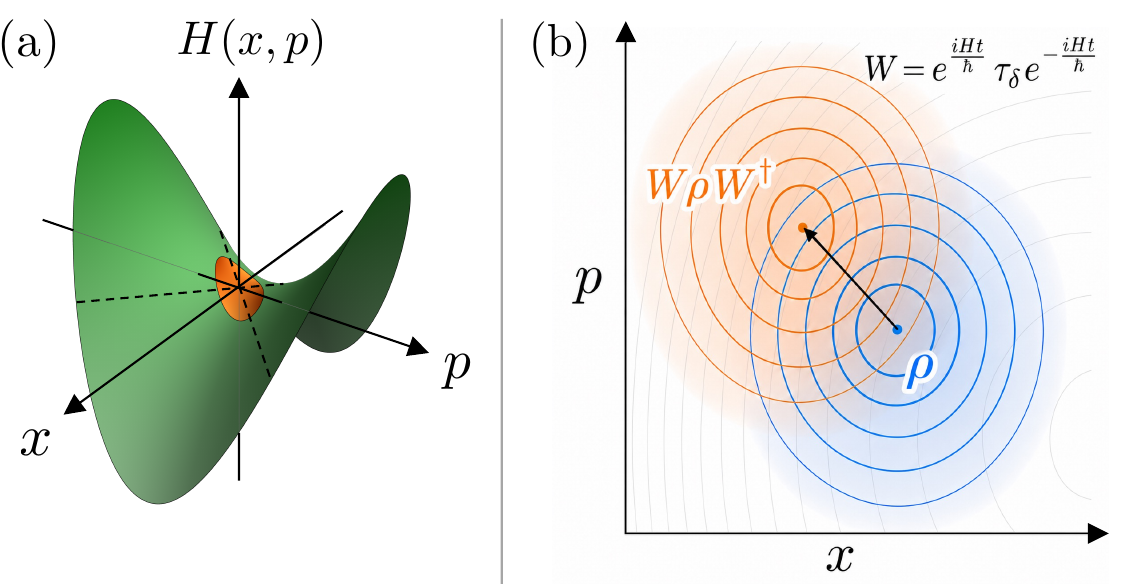}
    \caption{(a) Inverted harmonic oscillator Hamiltonian energy surface in phase space $x, p$. The initial distribution (orange) stretches exponentially along the diagonal and compresses exponentially perpendicular to it (dotted lines). (b) Part of the proof strategy for the $\lambda_{\rm q} \leq \lambda_{\rm exp}$ upper bound. A quantum state $\rho$ evolves under $W = e^{i H t} \tau_\delta e^{- i H t}$ (left).     The quantum optimal transport distance between $\rho$ and $W\rho W^\dagger$ is bounded in the companion work~\cite{cotler2025egorov} using classical coarse-graining together with phase-space localization estimates (right).}
    \label{fig:conceptual2}
\end{figure}

To place our framework in context, we compare our quantum Lyapunov exponent with the exponent extracted from out-of-time-order correlators (OTOCs), a widely studied diagnostic of quantum chaos~\cite{larkin1969quasiclassical, shenker2014black, maldacena2016bound, cotler2018out}.  OTOCs measure the growth of Heisenberg operator commutators and therefore provide an operator-level probe of sensitivity to perturbations.  In the phase space setting considered here, we define the associated exponent as follows:
\begin{definition}[OTOC Lyapunov exponent]
\label{def:OTOC}
Let $\{\rho_\hbar\}_{\hbar > 0}$ be a family of quantum states with a well-defined classical phase space limit, in the sense of Eq.~\eqref{eq:rho-sc-lim}. We define
\begin{equation*}
\lambda_{\mathrm{OTOC}} \equiv \limsup_{T \to \infty}\limsup_{\hbar \to 0} \frac{1}{2T}\log\!\left[\sigma_{\max}\!\left(\mathbf{C}_{\rho_\hbar}(T)\right)\right],
\end{equation*}
where $\mathbf{C}_{\rho_\hbar}(T)$ is the matrix with entries $C_{jk}(T) \equiv \frac{1}{\hbar^2}\operatorname{tr}\!\big(\rho_\hbar\left|[\hat r_j(T),\hat r_k(0)]\right|^2\big)$, and $\sigma_{\max}$ denotes the largest singular value.
\end{definition}

\noindent
\noindent
The OTOC probes a different aspect of the dynamics from the quantum optimal transport exponent. Our quantum optimal transport exponent $\lambda_{\rm q}$ optimizes the metric expansion over initial state perturbations, while the OTOC tracks the growth of specified phase-space operators against a fixed background state. For the class of Hamiltonians characterized in the Appendix, we show that the quantum optimal transport exponent equals the global classical expansion coefficient $\lambda_{\rm exp}$, while the OTOC exponent is bounded above by the same rate and saturates the bound when nearly maximally expanding trajectories are not exponentially suppressed by the limiting phase-space distribution.

\begin{theorem}[Main result, informal]
\label{thm:lambda-eq}
For dynamics obtained by quantizing a classical Hamiltonian system, and for any admissible quantum optimal transport distance,
\begin{equation}
\lambda_{\text{\rm q}} = \lambda_{\text{\rm exp}} \geq \lambda_{\text{\rm OTOC}}
\end{equation}
for every family of states $\rho_\hbar$ with a well-defined classical phase-space limit. The final inequality is saturated when the state assigns sufficient weight to trajectories exhibiting the maximal classical instability.
\end{theorem}
\noindent The precise version of this theorem, with complete assumptions on admissible Hamiltonians, is stated as Theorem~\ref{thm:precise-main} in the Appendix.

The equality $\lambda_{\rm q} = \lambda_{\rm exp}$ follows from the semiclassical comparison between quantum optimal transport and classical Wasserstein transport. Theorem~\ref{thm:upperbd1}, together with the Wasserstein--Egorov estimate proved by the authors in~\cite{cotler2025egorov} and stated in the Appendix, implies at fixed $T$ that
\[
d_p(\rho_\hbar(T),\sigma_\hbar(T)) \leq J(T)\,d_p(\rho_\hbar,\sigma_\hbar) + O(e^{CT}\hbar^{1/2}).
\]
The idea behind the proof of this inequality is to consider the time-evolved phase space translation operator $W_{\delta,t} := e^{-it\hat{H}/\hbar}\tau_\delta e^{it\hat{H}/\hbar}$,
and show that it is approximately $e^{Ct}\hbar^{1/2}$-local in phase space (see Figure~\ref{fig:conceptual2}).
Because the admissible initial separation satisfies $d_p(\rho_\hbar,\sigma_\hbar) \gg \hbar^{1/2}$, the correction vanishes relative to the initial distance as $\hbar\to0$, giving $\lambda_{\rm q}\leq\lambda_{\rm exp}$. For the reverse inequality, choose coherent states centered at phase-space points whose finite-time expansion approaches $J(T)$. Their initial quantum optimal transport distance equals the classical separation by the translation axiom, while semiclassical propagation shows that their evolved distance approaches the separation of the corresponding classical trajectories. This gives $\lambda_{\rm q}\geq\lambda_{\rm exp}$.

For the OTOC, Egorov's theorem gives at fixed $T$ the classical limit $C_{jk}^{\rm cl}(T) = \int |\{\Phi_T(\alpha)_j,r_k\}_{\rm PB}|^2\diff\mu(\alpha)$. Since the tangent map is bounded by the global Lipschitz factor $J(T)$, one has $\sigma_{\max}(\mathbf C^{\rm cl}(T)) \leq 2D\,J(T)^2$, and therefore $\lambda_{\rm OTOC}\leq\lambda_{\rm exp}$. The inequality is saturated when regions with nearly maximal finite-time expansion carry non-exponentially-small weight under the measure $\mu$.

Our result shows that quantum optimal transport captures the exponential metric instability of the underlying classical dynamics while remaining well defined for extended quantum states. The equality $\lambda_{\rm q} = \lambda_{\rm exp}$ identifies the global classical expansion rate as the semiclassical limit of the quantum optimal transport exponent; in the settings where $\lambda_{\rm exp} = \lambda_{\rm cl}$, this gives $\lambda_{\rm q} = \lambda_{\rm cl}$. By contrast, $\lambda_{\rm OTOC}$ is a state-dependent second-moment exponent. The OTOC exponent can be strictly smaller than $\lambda_{\rm exp}$ when nearly maximally expanding trajectories have exponentially small weight under the limiting phase-space distribution. Conversely, when such trajectories contribute without exponential suppression, $\lambda_{\rm OTOC} = \lambda_{\rm exp} = \lambda_{\rm q}$, and all three exponents coincide with $\lambda_{\rm cl}$ whenever $\lambda_{\rm exp} = \lambda_{\rm cl}$. Thus OTOCs provide a complementary operator-level probe of semiclassical instability, while quantum optimal transport measures the corresponding metric expansion of quantum states. This unified result demonstrates that quantum optimal transport may provide a rigorous foundation that connects diverse characterizations of quantum chaos.

\section{Discussion}

Our results establish a transport-geometric framework for quantum chaos. The admissible quantum distances incorporate the underlying phase-space geometry while remaining sensitive to genuinely quantum state structure, and the resulting quantum Lyapunov exponent recovers the classical global expansion coefficient in the semiclassical limit. In settings where $\lambda_{\rm exp} = \lambda_{\rm cl}$, the quantum Lyapunov exponent therefore also reproduces the usual maximal classical Lyapunov exponent. The accompanying comparison with OTOCs distinguishes this optimized metric instability from the state-weighted second-moment growth measured by operator commutators.

While our methodology leverages semiclassical reasoning so as to ensure conceptual consistency, it will be important going forward to explore phenomena that are more fully quantum. For example, future work might investigate quantum corrections to Lyapunov exponents through a systematic $\hbar$ expansion. Such corrections could reveal how quantum interference and entanglement modify chaotic dynamics beyond the leading semiclassical behavior. Understanding these corrections is particularly relevant for mesoscopic systems where $\hbar$ effects are non-negligible but the classical limit still provides a useful organizing principle~\cite{haake1991quantum}. More broadly, our framework suggests an approach for adapting classical chaos theory concepts~\cite{robinson1998dynamical}, such as attractors, mixing, and topological dynamics, to the quantum domain, potentially revealing entirely new structures that are unique to quantum mechanics.  For instance, it would be interesting to develop genuinely quantum analogues of Pesin's theorem, relating suitable notions of quantum Lyapunov growth to entropy production beyond the semiclassical regime.  Another natural direction is to formulate a quantum theory of uniformly hyperbolic dynamics, including quantum analogues of structural stability and shadowing.

A particularly promising direction involves quantum field theories and many-body spin systems, where fully characterizing chaos is challenging. Extending quantum optimal transport techniques to systems comprising qudit-based phase spaces~\cite{antonopoulos2025grand} could provide useful analytic tools for probing chaos and operator scrambling, and establish a more direct link to OTOCs.

Our work demonstrates that quantum optimal transport provides the missing geometric foundation for understanding chaos across quantum and classical regimes. To understand quantum chaos, we need to lean in to the counterintuitive nature of non-commutative quantum phase space and internalize its geometry.

\subsection*{Acknowledgements}

We would like to thank Luke Coffman and Semon Rezchikov for valuable conversations.  The authors used ChatGPT 5.6 Pro to discuss approaches to the functional-analytic subtleties involved in defining the Carlen–Maas distance on $L^2(\bbR^d)$, including the use of Petz monotonicity in the proof of Proposition~\ref{prp:Q-ub} in Appendix~\ref{sec:petz}.  The authors wrote the manuscript, carefully checked the mathematical arguments for correctness (as usual), and had ChatGPT 5.6 Pro assist with edits and typos. This work is supported by the U.S.~Department of Energy, Office of Science, under Award Number DE-SC0021013 titled `Emergent Phenomena in Quantum Dynamics: From Chaos to Spacetime'. JC is also supported by a fellowship from the Alfred P.~Sloan Foundation. This material is based upon work supported by the National Science Foundation under Grant No. DMS-2606659.

\bibliographystyle{apsrev4-1}
\bibliography{references}

\newpage
\onecolumngrid

\makeatletter
\let\set@footnotewidth\set@footnotewidth@one
\makeatother

\appendix

\makeatletter
\renewcommand{\thesubsection}{\thesection.\arabic{subsection}}
\renewcommand{\p@subsection}{}
\makeatother

\numberwithin{equation}{section}
\renewcommand{\theequation}{\Alph{section}.\arabic{equation}}
\renewcommand{\theHequation}{\Alph{section}.\arabic{equation}}

\appendixtableofcontents
\appendixtoctrue

\section{Notation and features of admissibility axioms}
\label{app:notation}

\subsection{General quantum notation}

Throughout the Appendix, $D$ denotes the number of degrees of freedom, so that the classical phase space is $\mathbb R^{2D}$. We work on the Hilbert space
\begin{equation}
\mathcal H := L^2(\mathbb R^D).
\end{equation}
We write $\mathcal B(\mathcal H)$ for the bounded operators on $\mathcal H$, $\mathfrak S_1(\mathcal H)$ for the trace-class operators, and $\mathfrak S_2(\mathcal H)$ for the Hilbert--Schmidt operators. The canonical trace is denoted by $\Tr(\cdot)$. If $A\in\mathcal B(\mathcal H)$ and $\xi\in\mathfrak S_1(\mathcal H)$, we use the duality pairing
\begin{equation}
\langle A,\xi\rangle := \Tr(A^\dagger\xi).
\end{equation}
When both entries are Hilbert--Schmidt, this agrees with the Hilbert--Schmidt inner product
\begin{equation}
\langle A,B\rangle_{\rm HS} := \Tr(A^\dagger B).
\end{equation}
We define
\begin{align}
\mathcal M_{\rm sa} &:= \{\xi\in\mathfrak S_1(\mathcal H):\xi=\xi^\dagger\},\\
\mathcal M_0 &:= \{\xi\in\mathcal M_{\rm sa}:\Tr\xi=0\},\\
\mathcal D &:= \{\rho\in\mathcal M_{\rm sa}:\rho\geq0,\ \Tr\rho=1\},\\
\mathcal D_+ &:= \{\rho\in\mathcal D:\ker\rho=\{0\}\}.
\end{align}
Thus, $\mathcal D$ is the space of density matrices and $\mathcal M_0$ is the space of trace-zero tangent vectors. The position and momentum operators are denoted by $\hat x_j$ and $\hat p_j$, and we use $\mathcal S(\mathbb R^D)$ as a common invariant core for the unbounded operators appearing below. The vector of phase-space operators is
\begin{equation}
\hat r := (\hat x_1,\ldots,\hat x_D,\hat p_1,\ldots,\hat p_D).
\end{equation}
Whenever a commutator with an unbounded operator appears, it is first understood on $\mathcal S(\mathbb R^D)$. If the commutator is used in an operator norm, a Lipschitz seminorm, or the Carlen--Maas Dirichlet form~\cite{carlen2020non}, we require that it extend to a bounded operator on $\mathcal H$.

\subsection{Symplectic conventions and Weyl translations}

Identifying phase space with $\mathbb R^{2D}$, a phase-space point is denoted by $\alpha=(q,p)$, with $q,p\in\mathbb R^D$. The Euclidean inner product and norm are denoted by $\alpha\cdot\beta$ and $|\alpha|$. The symplectic matrix is
\begin{equation}
J := \begin{pmatrix} 0 & -I_D \\ I_D & 0 \end{pmatrix}\,,
\end{equation}
and with this convention,
\begin{equation}
[\hat r_j,\hat r_k] = -i\hbar J_{jk}\mathds 1.
\end{equation}
For $\alpha\in\mathbb R^{2D}$, define the self-adjoint linear observable
\begin{equation}
\hat X_\alpha := \alpha^\intercal J\hat r = p\cdot\hat x-q\cdot\hat p
\end{equation}
which satisfies
\begin{equation}
[\hat X_\alpha,\hat r_j] = i\hbar\alpha_j\mathds 1.
\end{equation}
The Weyl unitary associated with $\alpha$ is
\begin{equation}
\tau_\alpha := \exp\!\left(\frac{i}{\hbar}\hat X_\alpha\right)
\end{equation}
and it satisfies
\begin{equation}
\tau_\alpha^\dagger = \tau_{-\alpha}\,,
\qquad
\tau_\alpha\tau_\beta = \exp\!\left(\frac{i}{2\hbar}\alpha^\intercal J\beta\right)\tau_{\alpha+\beta}\,.
\end{equation}
We define conjugation by the Weyl unitary as
\begin{equation}
\mathcal T_\alpha[X] := \tau_\alpha X\tau_\alpha^\dagger\,,
\qquad
\mathcal T_\alpha^*[A] := \tau_\alpha^\dagger A\tau_\alpha\,.
\end{equation}
On states, $\mathcal T_\alpha$ is the translation channel, while $\mathcal T_\alpha^*$ is its Heisenberg adjoint. With these conventions,
\begin{equation}
\mathcal T_\alpha^*[\hat r] = \hat r+\alpha\mathds 1\,,
\end{equation}
and consequently
\begin{equation}
\Tr\!\left(\mathcal T_\alpha[\rho]\hat r\right) = \Tr(\rho\hat r)+\alpha
\end{equation}
whenever the first moment is finite. We denote the corresponding classical translation by
\begin{equation}
\vartheta_\alpha(\beta) := \beta+\alpha\,,
\end{equation}
and thus $(\vartheta_\alpha)_\#\mu$ denotes the pushforward of a classical measure under translation by $\alpha$.

Given $S\in\operatorname{Sp}(2D,\mathbb R)$, let $U_S$ be a metaplectic unitary satisfying
\begin{equation}
\label{eq:US-unitary}
U_S^\dagger\hat r\,U_S = S\hat r.
\end{equation}
We define
\begin{equation}
\label{eq:US-map}
\mathcal U_S[X] := U_S XU_S^\dagger\,,
\qquad
\mathcal U_S^*[A] := U_S^\dagger A U_S\,,
\end{equation}
and so $\mathcal U_S$ acts on states by Schr\"{o}dinger-picture conjugation and $\mathcal U_S^*$ is its Heisenberg adjoint. They satisfy
\begin{equation}
\Tr\!\left(\mathcal U_S[\rho]A\right) = \Tr\!\left(\rho\,\mathcal U_S^*[A]\right).
\end{equation}
For matrices acting on $\mathbb R^{2D}$, $\|S\|$ denotes the operator norm induced by the Euclidean norm. We reserve $\|\cdot\|_{\rm op}$ for the operator norm on $\mathcal H$.

\subsection{Harmonic oscillator, number basis, and coherent states}

To define the quantum analogue of the $W_2$ distance, we need the analogue of the squared displacement from the origin of a state.  The natural quantum analogue of squared phase-space displacement is provided by the harmonic oscillator Hamiltonian, which also supplies the ladder operators used in the Carlen--Maas construction. We take a reference harmonic oscillator Hamiltonian to be
\begin{align}
\hat H_0 & := \frac{1}{2}\sum_{j = 1}^D(\hat x_j^2 + \hat p_j^2). \nonumber
\end{align}
For each spatial index $j = 1, \ldots, D$, we define the ladder operators
\begin{align}
\hat a_j & := (2\hbar)^{-1/2}(\hat x_j + \rmi \hat p_j)\,, \quad
\hat a_j^\dagger := (2\hbar)^{-1/2}(\hat x_j - \rmi \hat p_j)\,. \nonumber
\end{align}
These satisfy the usual relations
\begin{align}
[\hat a_j, \hat a_k^\dagger] & = \delta_{jk}\One\,, \quad
[\hat a_j, \hat a_k] = 0\,, \quad
[\hat a_j^\dagger, \hat a_k^\dagger] = 0\,. \nonumber
\end{align}
Equivalently,
\begin{align}
\hat x_j & = \sqrt{\frac{\hbar}{2}}(\hat a_j + \hat a_j^\dagger)\,, \quad
\hat p_j = - \rmi\sqrt{\frac{\hbar}{2}}(\hat a_j - \hat a_j^\dagger)\,. \nonumber
\end{align}
The number operators are then
\begin{align}
\hat N_j := \hat a_j^\dagger \hat a_j\,,\,\quad
\hat N := \sum_{j = 1}^D \hat N_j\,. \nonumber
\end{align}
With these definitions we can write $\hat{H}_0$ as
\begin{align}
\hat H_0 & = \hbar\left(\hat N + \frac{D}{2}\One\right)\,. \nonumber
\end{align}

The joint number basis is denoted by $|\mathbf n\rangle=|n_1,\ldots,n_D\rangle$ for $\mathbf n\in\mathbb N^D$, with $\hat N_j|\mathbf n\rangle=n_j|\mathbf n\rangle$. We write $|\mathbf n|:=n_1+\cdots+n_D$, and let $|0\rangle:=|0,\ldots,0\rangle$ denote the oscillator vacuum. The coherent state centered at $\alpha\in\mathbb R^{2D}$ is
\begin{equation}
|\alpha\rangle := \tau_\alpha|0\rangle.
\end{equation}
With the translation convention above,
\begin{equation}
\langle\alpha|\hat r|\alpha\rangle = \alpha.
\end{equation}
The coherent states resolve the identity:
\begin{equation}
\int_{\mathbb R^{2D}}|\alpha\rangle\langle\alpha|\,\frac{\diff\alpha}{(2\pi\hbar)^D} = \mathds 1.
\end{equation}
For a trace-class operator $\xi$, define its Husimi distribution by
\begin{equation}
\Husimi_\xi(\alpha) := \frac{1}{(2\pi\hbar)^D}\langle\alpha|\xi|\alpha\rangle\,,
\end{equation}
and so if $\rho\in\mathcal D$, then we have
\begin{equation}
\int_{\mathbb R^{2D}}\Husimi_\rho(\alpha)\diff\alpha = 1.
\end{equation}
The Husimi distribution is covariant under translations:
\begin{equation}
\Husimi_{\mathcal T_a[\rho]}(\alpha) = \Husimi_\rho(\alpha-a).
\end{equation}
Whenever the first moment exists,
\begin{equation}
\Tr(\rho\hat r) = \int_{\mathbb R^{2D}}\alpha\,\Husimi_\rho(\alpha)\diff\alpha.
\end{equation}
We denote the Husimi measurement map by
\begin{equation}
\mathcal C[\rho] := \Husimi_\rho\,,
\end{equation}
and its adjoint is the so-called anti-Wick quantization,
\begin{equation}
\label{eq:anti-Wick}
\mathcal C^*[f] := \int_{\mathbb R^{2D}}f(\alpha)|\alpha\rangle\langle\alpha|\,\frac{\diff\alpha}{(2\pi\hbar)^D}\,,
\end{equation}
and the two maps satisfy
\begin{equation}
\label{eq:C*Husimi-dual}
\Tr\!\left(\rho\,\mathcal C^*[f]\right) = \int_{\mathbb R^{2D}}f(\alpha)\Husimi_\rho(\alpha)\diff\alpha.
\end{equation}
The coherent-state noising channel is
\begin{equation}
\label{eq:noising-channel}
\mathcal N[\rho] := \int_{\mathbb R^{2D}}\Husimi_\rho(\alpha)|\alpha\rangle\langle\alpha|\diff\alpha.
\end{equation}

Finally, for $p\geq1$, define
\begin{equation}
\mathcal D_p := \left\{\rho\in\mathcal D:\int_{\mathbb R^{2D}}|\alpha|^p\Husimi_\rho(\alpha)\diff\alpha<\infty\right\}.
\end{equation}
For probability measures $\mu$ and $\nu$ on $\mathbb R^{2D}$ with finite $p$th moments, let $\Pi(\mu,\nu)$ denote the set of couplings of $\mu$ and $\nu$. The classical $p$-Wasserstein distance is
\begin{equation}
W_p(\mu,\nu)^p := \inf_{\pi\in\Pi(\mu,\nu)} \int_{\mathbb R^{2D}\times\mathbb R^{2D}} |\alpha-\alpha'|^p\diff\pi(\alpha,\alpha').
\end{equation}

\subsection{Quantization and symbol classes}
The analysis of $\lambda_{\rm OTOC}$ uses general results from semiclassical analysis.  
For smooth functions $a\in C^\infty(\bbR^{2D})$, we define the Weyl quantization
\begin{equation}
\bigl(\Op_\hbar(a)f\bigr)(x) := \frac{1}{(2\pi\hbar)^D}\int_{\mathbb R^{2D}}e^{i(x-y)\cdot p/\hbar}a\!\left(\frac{x+y}{2},p\right)f(y)\diff y\diff p.
\end{equation}
More precisely, we consider quantization of functions with bounded derivatives of sufficiently high order, and define
\begin{equation}
\label{eq:Sk-def}
S_k(1) := \{a\in C^\infty(\bbR^{2D}) \mid \text{For all }|\alpha|\geq k, \sup_{r\in\bbR^{2D}} |\partial^\alpha a(r)| < C_\alpha\}.
\end{equation}
Thus functions $a\in S_k(1)$ have uniformly bounded derivatives of order at least $k$.   We use simply $S(1)$ below to mean $S_0(1)$,
and we also allow for symbols that have an implicit dependence on $\hbar$, so long as the constants $C_\alpha$ are uniform in $\hbar$.
The Calder\'on-Vaillancourt theorem ensures that $\Op_\hbar(a)$ is bounded as an operator on $L^2(\bbR^D)$ 
when $a\in S_0(1)$.   

For $a\in S_k(1)$ and $b\in S_j(1)$, the product $\Op_\hbar(a)\Op_\hbar(b)$ has symbol expansion
\[
a\#b = ab + \frac{i\hbar}{2} \{a,b\}_{\rm PB} + O(\hbar^2),
\]
where $\{\cdot,\cdot\}_{\rm PB}$ is the Poisson bracket and the error is measured in $S_{\max\{0,k-2,j-2\}}(1)$.  In particular, we have the product formula
\begin{equation}
\Op_\hbar(a)\Op_\hbar(b) = \Op_\hbar(ab) + O(\hbar),
\end{equation}
and the commutator rule
\begin{equation}
\frac{1}{i\hbar}[\Op_\hbar(a),\Op_\hbar(b)] = \Op_\hbar(\{a,b\}_{\rm PB}) + O(\hbar).
\end{equation}
In fact, this sharpens to an exact identity if $a$ or $b$ is a linear function in phase space.  

Egorov's theorem describes the time evolution of such operators.  The relevant assumption on the 
Hamiltonians is $H\in S_2(1)$, which corresponds to the Hamiltonian vector field having bounded derivatives.  This class of Hamiltonians is used throughout the rest of this work, so we define it here.
\begin{definition}[Admissible Hamiltonians]
\label{def:admissible-H}
A Hamiltonian function $H\in C^\infty(\bbR^{2D})$ is said to be \emph{admissible} if $H\in S_2(1)$, that is if for all multi-indices $|\alpha|\geq 2$, there is a constant $C_\alpha$ such that
\[
\sup_{\alpha\in\bbR^{2D}} |\partial^\alpha H(\alpha)| < C_\alpha.
\]
\end{definition}
Physically, this condition excludes singular or arbitrarily rapidly varying forces: since the Hessian of $H$ is uniformly bounded, the Hamiltonian vector field $X_H = J\nabla H$ is globally Lipschitz and nearby trajectories have a finite expansion factor at every fixed time. The bounds on higher derivatives provide the regularity needed for the semiclassical estimates below. This class includes, for example, quadratic Hamiltonians and smooth compactly supported perturbations thereof.

We take the observables to lie in $S_1(1)$, which is sufficiently broad to include the phase-space coordinate observables used in the OTOC definition. We state the version of Egorov's  theorem we need which is due to Bouzouina and Robert~\cite{BouzouinaRobert2002}.  
\begin{theorem}[Consequence of Theorem 1.2 of~\cite{BouzouinaRobert2002}]
Let $H\in S_2(1)$ be a Hamiltonian with bounded derivatives of order 2 and higher, and let $A\in S_1(1)$
have bounded gradients of first order and higher.  Then
\[
e^{it\hat{H}/\hbar}\hat{A} e^{-it\hat{H}/\hbar} = \Op_\hbar(\Phi_t A) + \hbar \,S_1(1),
\]
where the term on the right indicates that this is up to an error that is uniformly bounded by $\hbar$
in the $S_1(1)$ seminorms.
In particular, for $A,B\in S_1(1)$, we have
\[
-i\hbar^{-1}[\hat{A}(t), \hat{B}(0)] = \Op_\hbar(\{\Phi_t A, B\}_{\rm PB}) + \hbar \,S_0(1).
\]
\end{theorem}

For an admissible Hamiltonian $H$, let $\Phi_T: \bbR^{2D} \to \bbR^{2D}$ denote the corresponding classical flow. Egorov's theorem states that for fixed $T$ and $a\in S(1)$,
\begin{equation}
e^{iT\hat H/\hbar}\Op_\hbar(a)e^{-iT\hat H/\hbar} = \Op_\hbar(a\circ\Phi_T) + O(\hbar).
\end{equation}
The other ingredient we need is the notion of a semiclassical limit. Let $\{\rho_\hbar\}_{\hbar>0}$ be a family of quantum states. We say that this family has the classical phase-space limit $\mu$ if, for every $a\in S(1)$,
\begin{equation}
\label{eq:rho-sc-lim}
\lim_{\hbar\to0}\Tr\!\left[\Op_\hbar(a)\rho_\hbar\right] = \int_{\mathbb R^{2D}}a(\alpha)\diff\mu(\alpha).
\end{equation}

Another quantization we use is the ``coherent state'' quantization $\Op_\hbar^{\rm coh}$,
\begin{equation}
\label{eq:coh-quant}
\Op_\hbar^{\rm coh}(a) := (2\pi \hbar)^{-D} \int \ket{\alpha}\bra{\alpha}a(\alpha)\diff\alpha,
\end{equation}
which is the same as the anti-Wick quantization $\mcal{C}^*$ described above.  
In this notation, the duality~\eqref{eq:C*Husimi-dual} is written
\begin{equation}
\label{eq:coh-duality}
\Tr [ \Op_\hbar^{\rm coh}(a) \rho] = \int a(\alpha) \Husimi_\rho(\alpha)\diff\alpha.
\end{equation}
This quantization has the boundedness property
\begin{equation}
\label{eq:coh-bd}
\|\Op_\hbar^{\rm coh}(a)\|_{\rm op} \leq \|a\|_{L^\infty},
\end{equation}
and also the exact commutator identity for testing against linear functions,
\begin{equation}
\label{eq:coh-comm}
-\frac{i}{\hbar}
[\hat X_\alpha,\Op_\hbar^{\rm coh}(a)] = \Op_\hbar^{\rm coh}(\alpha\cdot\nabla a).
\end{equation}

\subsection{Test observables and Lipschitz classes}
\label{sec:Lipsec}
We reserve hats for the canonical unbounded operators, such as $\hat x_j$, $\hat p_j$, $\hat r_j$, $\hat a_j$, and $\hat a_j^\dagger$. General observables are denoted by $A$, $B$, or $F$. The $W_1$ distance uses the canonical Lipschitz seminorm
\begin{equation}
\label{eq:Lipnorm-def}
\|A\|_{\rm Lip} := \sup_{|u| = 1}\frac{1}{\hbar}\|[u \cdot \hat r, A]\|_{\rm op}
\end{equation}
whenever the commutators extend to bounded operators. The factor $\hbar^{-1}$ is included so that linear observables have their usual classical Lipschitz constants. Indeed, if $A = v \cdot \hat r$, then
\begin{align}
\|A\|_{\rm Lip} & = |v|. \nonumber
\end{align}

For the Carlen--Maas 2-Wasserstein construction~\cite{carlen2020non}, the relevant noncommutative derivatives are commutators with the jump operators $V_\ell$ defined below. We write $\mcal B_0$ for a real vector space of self-adjoint test observables $A$ such that $A$ preserves the Schwartz core and each commutator $[V_\ell, A]$ extends to a bounded operator. Constants are harmless in the dual action because tangent vectors have trace zero. The class $\mcal B_0$ is chosen large enough to contain the linear observables $\hat X_\alpha$ and stable enough under the finite-energy cutoffs $A \mapsto \Pi_N A\Pi_N$ used in approximation arguments.

\subsection{Ornstein--Uhlenbeck and Carlen--Maas jump notation}

For a jump operator $L$, we write the Lindblad dissipator as
\begin{align}
\msf D[L]\rho & := L\rho L^\dagger - \frac{1}{2}\bigl(L^\dagger L\rho + \rho L^\dagger L\bigr). \nonumber
\end{align}
The symbol $\msf D[L]$ denotes a dissipator and should not be confused with the state space $\mcal D$.  The Carlen--Maas transport geometry we use is built from the harmonic-oscillator ladder operators. We introduce the index set
\begin{align}
\mcal I & := \{(j, -), (j, +) : j = 1, \ldots, D\}. \nonumber
\end{align}
For $\ell = (j, \pm) \in \mcal I$, define
\begin{align}
V_{j, -} & := \hbar^{-1/2}\hat a_j\,,\quad 
V_{j, +}  := \hbar^{-1/2}\hat a_j^\dagger\,. \nonumber
\end{align}
We also define the corresponding frequencies
\begin{align}
\omega_{j, -} & := - 2\beta\hbar\, \quad \omega_{j, +} := 2\beta\hbar\,. \nonumber
\end{align}
When the index is not important, we write $V_\ell$ and $\omega_\ell$.

With these conventions, the predual Ornstein--Uhlenbeck generator acting on states is
\begin{align}
\mcal L_\beta^\dagger\rho & := \frac{2}{\hbar}\sum_{j = 1}^D\left(e^{\beta\hbar}\msf D[\hat a_j]\rho + e^{-\beta\hbar}\msf D[\hat a_j^\dagger]\rho\right) = 2\sum_{\ell \in \mcal I} e^{-\omega_\ell/2}\msf D[V_\ell]\rho. \nonumber
\end{align}
The parameter convention here is tied to the choice $\omega_{j, \pm} = \pm 2\beta\hbar$. With this convention, the Ornstein--Uhlenbeck generator has the stationary Gibbs state
\begin{align}
\sigma_\beta &:= \frac{1}{Z_\beta}e^{-2\beta\hat H_0}, \qquad Z_\beta := \Tr(e^{-2\beta\hat H_0})\,, \nonumber
\end{align}
for $\beta>0$, so that $\mcal L_\beta^\dagger[\sigma_\beta]=0$. If one instead wants the Gibbs state to be written as $\sigma_\beta = e^{-\beta\hat H_0}/\Tr(e^{-\beta\hat H_0})$, then all frequencies $\omega_{j, \pm}$ above should be replaced by $\pm\beta\hbar$.

\subsection{Some examples of quantum optimal transport distances and their relation to the axioms}
\label{subsec:elementary-examples}

Before turning to the Carlen--Maas distances, we record several elementary constructions that help clarify the role of each of the four axioms from the main text. These examples distinguish the structural requirements of translation covariance and convexity from the information-processing condition, and illustrate what the latter excludes.

Recall that we use the coherent-state normalization for which $\int_{\bbR^{2D}} \Husimi_\rho(\alpha) \diff\alpha = 1$, where $\Husimi_\rho(\alpha) = (2\pi\hbar)^{-D}\langle\alpha|\rho|\alpha\rangle$. We restrict throughout to states for which the moments below are finite. Define the phase-space first moment by
\begin{equation}
\bar r(\rho) := \Tr(\rho\hat r).
\end{equation}
With our translation convention,
\begin{align}
\bar r(\mcal T_a[\rho]) = \bar r(\rho) + a, \qquad \Husimi_{\mcal T_a[\rho]}(\alpha) = \Husimi_\rho(\alpha-a). \nonumber \end{align}
The norm $|\cdot|$ below is the same Euclidean phase-space norm appearing in the translation-distance axiom. If one instead formulates the axioms using an $\ell^p$ ground norm on phase space, the same discussion applies after replacing $|\cdot|$ by $\|\cdot\|_{\ell^p}$.

\begin{example}[First-moment pseudodistance]
For $p \geq 1$, define
\begin{align} \label{eq:first-moment-pseudodistance} d_{\mathrm{mean},p}(\rho,\sigma) := |\bar r(\rho) - \bar r(\sigma)|. \end{align}
This is generally only a pseudometric, since distinct states with the same first moments have zero distance. Nevertheless, it satisfies the first three structural axioms.

Translation invariance follows from
\begin{align} \bar r(\mcal T_a[\rho]) - \bar r(\mcal T_a[\sigma]) = \bar r(\rho) - \bar r(\sigma). \nonumber \end{align}
The translation normalization follows from
\begin{align} d_{\mathrm{mean},p}(\rho,\mcal T_a[\rho]) = |\bar r(\rho) - \bar r(\rho) - a| = |a|. \nonumber \end{align}
Finally, suppose $\rho = \int \rho_\theta \diff\mu(\theta)$ and $\sigma = \int \sigma_\phi \diff\nu(\phi)$, and let $\Gamma$ be any coupling of $\mu$ and $\nu$. Since
\begin{align} \bar r(\rho) - \bar r(\sigma) = \int \big(\bar r(\rho_\theta) - \bar r(\sigma_\phi)\big) \diff\Gamma(\theta,\phi), \nonumber \end{align}
Jensen's inequality for the convex function $v \mapsto |v|^p$ gives
\begin{align} d_{\mathrm{mean},p}(\rho,\sigma)^p & = \left|\int \big(\bar r(\rho_\theta) - \bar r(\sigma_\phi)\big) \diff\Gamma(\theta,\phi)\right|^p \nonumber\\ &\leq \int |\bar r(\rho_\theta) - \bar r(\sigma_\phi)|^p \diff\Gamma(\theta,\phi) \nonumber\\ & = \int d_{\mathrm{mean},p}(\rho_\theta,\sigma_\phi)^p \diff\Gamma(\theta,\phi). \nonumber \end{align}

The information-processing axiom, however, fails in general. Since the Husimi distribution reproduces first moments,
\begin{align}
\bar r(\rho) = \int_{\bbR^{2D}}\alpha\,\Husimi_\rho(\alpha)\diff\alpha. \nonumber
\end{align}
Every coupling $\pi$ of $\Husimi_\rho$ and $\Husimi_\sigma$ therefore satisfies
\begin{align}
d_{\mathrm{mean},p}(\rho,\sigma)^p &= \left|\int_{\bbR^{2D}\times\bbR^{2D}}(\alpha-\alpha')\diff\pi(\alpha,\alpha')\right|^p \nonumber\\
&\leq \int_{\bbR^{2D}\times\bbR^{2D}}|\alpha-\alpha'|^p\diff\pi(\alpha,\alpha'). \nonumber
\end{align}
Infimizing over $\pi$ gives
\begin{align}
\label{eq:mean-below-Husimi}
d_{\mathrm{mean},p}(\rho,\sigma) \leq W_p(\Husimi_\rho,\Husimi_\sigma),
\end{align}
whereas the information-processing axiom requires the reverse inequality.

For example, let $\rho = |0\rangle\langle 0|$ be the oscillator vacuum and fix a nonzero vector $a \in \mathbb R^{2D}$. Define
\begin{equation}
\sigma := \frac{1}{2}\left(\mathcal T_a[\rho] + \mathcal T_{-a}[\rho]\right).
\end{equation}
Then $\bar r(\rho) = \bar r(\sigma) = 0$, and hence $d_{\mathrm{mean},p}(\rho,\sigma) = 0$. Nevertheless, $\rho \neq \sigma$, since $\rho$ is pure whereas $\sigma$ is a nontrivial mixture of two distinct coherent states. Their Husimi distributions are therefore distinct, and
\begin{equation}
W_p(\Husimi_\rho,\Husimi_\sigma) > 0.
\end{equation}
Thus, the first-moment construction fails both metric nondegeneracy and the information-processing axiom.
\end{example}

\begin{example}[Husimi--Wasserstein distance]
Let $W_p$ denote the classical $p$-Wasserstein distance on $\bbR^{2D}$ with cost $|\alpha-\alpha'|^p$. Define
\begin{align}
\label{eq:Husimi-Wasserstein-distance}
d_{\mathrm{H},p}(\rho,\sigma) := W_p(\Husimi_\rho,\Husimi_\sigma).
\end{align}
This is the pullback of the classical Wasserstein distance along the Husimi measurement map, closely related to the Husimi-based optimal transport distances introduced in~\cite{zyczkowski1998monge,zyczkowski2001monge}. On the finite-$p$-moment domain, it is a genuine metric, since the coherent-state POVM is informationally complete.

The first three axioms follow from the corresponding elementary properties of classical Wasserstein distance. Translation covariance of the Husimi distribution gives
\begin{align} d_{\mathrm{H},p}\big(\mcal T_a[\rho],\mcal T_a[\sigma]\big) & = W_p\big((\vartheta_a)_\#\Husimi_\rho,(\vartheta_a)_\#\Husimi_\sigma\big) \nonumber\\ & = W_p(\Husimi_\rho,\Husimi_\sigma), \nonumber \end{align}
where $\vartheta_a(\alpha) = \alpha+a$. Moreover, for any probability measure $\mu$ with finite $p$-moment,
\begin{align} \label{eq:Wp-translation-exact} W_p\big(\mu,(\vartheta_a)_\#\mu\big) = |a|. \end{align}
Indeed, the coupling $\alpha\mapsto\alpha+a$ gives $W_p\bigl(\mu,(\vartheta_a)_\#\mu\bigr)\leq|a|$. Conversely, for any coupling $\pi$ of $\mu$ and $(\vartheta_a)_\#\mu$, Jensen's inequality gives
\begin{align}
\int_{\bbR^{2D}\times\bbR^{2D}}|\alpha-\alpha'|^p\diff\pi(\alpha,\alpha') &\geq \left|\int_{\bbR^{2D}\times\bbR^{2D}}(\alpha-\alpha')\diff\pi(\alpha,\alpha')\right|^p \nonumber\\
&= |a|^p. \nonumber
\end{align}
Taking the infimum over $\pi$ proves Eq.~\eqref{eq:Wp-translation-exact}. Hence
\begin{align}
d_{\mathrm{H},p}(\rho,\mcal T_a[\rho]) = |a|. \nonumber 
\end{align}

It remains to check double convexity. Let $\rho = \int \rho_\theta \diff\mu(\theta)$ and $\sigma = \int \sigma_\phi \diff\nu(\phi)$, and let $\Gamma$ be a coupling of $\mu$ and $\nu$. By linearity of the Husimi map,
\begin{align} 
\Husimi_\rho & = \int \Husimi_{\rho_\theta} \diff\mu(\theta), \qquad \Husimi_\sigma = \int \Husimi_{\sigma_\phi} \diff\nu(\phi). \nonumber \end{align}
For finite mixtures, choose optimal couplings between each pair $\Husimi_{\rho_\theta}$ and $\Husimi_{\sigma_\phi}$ and average these couplings with respect to $\Gamma$. This produces a coupling between $\Husimi_\rho$ and $\Husimi_\sigma$, and therefore
\begin{align} d_{\mathrm{H},p}(\rho,\sigma)^p & = W_p(\Husimi_\rho,\Husimi_\sigma)^p \nonumber\\ &\leq \int W_p\big(\Husimi_{\rho_\theta},\Husimi_{\sigma_\phi}\big)^p \diff\Gamma(\theta,\phi) \nonumber\\ & = \int d_{\mathrm{H},p}(\rho_\theta,\sigma_\phi)^p \diff\Gamma(\theta,\phi). \nonumber \end{align}
The general case follows by approximation by simple functions and lower semicontinuity of $W_p$.

Finally, the information-processing axiom holds with equality:
\begin{align} W_p(\Husimi_\rho,\Husimi_\sigma) = d_{\mathrm{H},p}(\rho,\sigma). \nonumber \end{align}
Thus, $d_{\mathrm{H},p}$ satisfies all four axioms.
\end{example}

\begin{remark}[Interpolating between the two examples]
The previous two constructions may be interpolated. For $0 \leq \eta \leq 1$, define
\begin{align} \label{eq:elementary-interpolation} d_{\eta,p}(\rho,\sigma) := \left((1 - \eta)\,d_{\mathrm{mean},p}(\rho,\sigma)^p + \eta \,d_{\mathrm{H},p}(\rho,\sigma)^p\right)^{1/p}. \end{align}
For $\eta = 0$, this reduces to the first-moment pseudodistance. For $0 < \eta \leq 1$, it is a genuine metric, since $d_{\mathrm{H},p}$ is a genuine metric. The triangle inequality follows from the usual Minkowski inequality for the $\ell^p$-sum of two pseudometrics.

Translation invariance is immediate from translation invariance of the two summands. The translation normalization is exact because both summands give $|a|$:
\begin{align} d_{\eta,p}(\rho,\mcal T_a[\rho]) = \left((1 - \eta)|a|^p + \eta|a|^p\right)^{1/p} = |a|. \nonumber \end{align}
Double convexity follows by applying double convexity to the two summands separately:
\begin{align} d_{\eta,p}(\rho,\sigma)^p & = (1 - \eta)d_{\mathrm{mean},p}(\rho,\sigma)^p + \eta d_{\mathrm{H},p}(\rho,\sigma)^p \nonumber\\ &\leq \int \left((1 - \eta)d_{\mathrm{mean},p}(\rho_\theta,\sigma_\phi)^p + \eta d_{\mathrm{H},p}(\rho_\theta,\sigma_\phi)^p\right) \diff\Gamma(\theta,\phi) \nonumber\\ & = \int d_{\eta,p}(\rho_\theta,\sigma_\phi)^p \diff\Gamma(\theta,\phi). \nonumber \end{align}

Equation~\eqref{eq:mean-below-Husimi} implies
\begin{align} d_{\eta,p}(\rho,\sigma) \leq d_{\mathrm{H},p}(\rho,\sigma)\,,\nonumber
\end{align}
although the information-processing axiom requires the reverse inequality. Consequently, the two inequalities can hold for every pair of states only if $d_{\eta,p} = d_{\mathrm{H},p}$. In particular, for any pair of distinct states with the same first moments,
\begin{align} d_{\eta,p}(\rho,\sigma) = \eta^{1/p}d_{\mathrm{H},p}(\rho,\sigma). \nonumber \end{align}
Thus, $d_{\eta,p}$ fails the information-processing axiom whenever $0 \leq \eta < 1$. Only the endpoint $\eta = 1$, corresponding to the Husimi--Wasserstein distance, satisfies all four axioms.
\end{remark}
\vspace{3pt}
These examples clarify the roles of the four axioms. The first three axioms are compatibility conditions with affine phase-space geometry and convex mixing, but they permit distances that retain only a small part of the phase-space information. The information-processing axiom rules out such artificial transport shortcuts by requiring the quantum distance to dominate the classical Wasserstein distance between Husimi distributions. It does not, however, enforce a genuinely noncommutative transport geometry, since the Husimi--Wasserstein distance satisfies all four axioms while being obtained by first applying a commutative measurement and then using ordinary classical optimal transport.

\begin{proposition}[Classical analogue]
Let $\delta_p$ be a distance on probability measures on $\mathbb R^{2D}$ with finite $p$th moments. Suppose that $\delta_p$ satisfies the classical analogues of translation invariance, exact normalization on translations, and double convexity. Then
\begin{equation}
\delta_p(\mu,\nu)\leq W_p(\mu,\nu).
\end{equation}
If in addition
\begin{equation}
W_p(\mu,\nu)\leq\delta_p(\mu,\nu),
\end{equation}
which is analogous to the fourth axiom, then $\delta_p = W_p$.
\end{proposition}

\begin{proof}
Let $\pi$ be any coupling of $\mu$ and $\nu$. Using the decompositions of $\mu$ and $\nu$ into Dirac masses and applying double convexity,
\begin{align}
\delta_p(\mu,\nu)^p \leq \int_{\bbR^{2D}\times\bbR^{2D}}\delta_p(\delta_\alpha,\delta_{\alpha'})^p\diff\pi(\alpha,\alpha'). \nonumber
\end{align}
Since
\begin{equation}
\delta_{\alpha'} = (\vartheta_{\alpha'-\alpha})_\#\delta_\alpha,
\end{equation}
exact normalization on translations gives
\begin{equation}
\delta_p(\delta_\alpha,\delta_{\alpha'}) = |\alpha-\alpha'|.
\end{equation}
Therefore,
\begin{align}
\delta_p(\mu,\nu)^p \leq \int_{\bbR^{2D}\times\bbR^{2D}}|\alpha-\alpha'|^p\diff\pi(\alpha,\alpha'). \nonumber
\end{align}
Infimizing over $\pi$ proves $\delta_p(\mu,\nu)\leq W_p(\mu,\nu)$. Combining this with the assumed reverse inequality gives $\delta_p = W_p$.
\end{proof}

The first three classical axioms do not by themselves determine $W_p$. For example, the barycenter pseudodistance
\begin{align} \label{eq:classical-barycenter-pseudodistance} d_{\mathrm{bar},p}^{\rm cl}(\mu,\nu) := \left|\int_{\bbR^{2D}} \alpha\diff\mu(\alpha) - \int_{\bbR^{2D}} \alpha\diff\nu(\alpha)\right| \end{align}
satisfies the classical analogues of translation invariance, translation normalization, and double convexity, but is plainly not the Wasserstein distance. As above, Jensen's inequality gives
\begin{align} d_{\mathrm{bar},p}^{\rm cl}(\mu,\nu) \leq W_p(\mu,\nu). \nonumber \end{align}
Thus, $W_p$ is maximal among classical distances satisfying the first three axioms. Imposing the classical analogue of the fourth axiom supplies the reverse inequality and therefore uniquely selects $W_p$.

The situation is different in the quantum setting. The fourth axiom provides a lower bound relative to the fixed classical Husimi--Wasserstein geometry, but it does not determine how distances should behave in genuinely noncommutative directions. Accordingly, an intrinsic characterization of noncommutative quantum optimal transport would require an additional structural condition beyond the four axioms. A natural candidate is a noncommutative convex-duality requirement: the distance should arise from a cotangent Dirichlet form $Q_\rho(A)$, built from commutators with distinguished quantum gradients, whose Legendre transform defines a tangent action and whose induced length distance is the transport metric. We do not attempt to formulate such an axiom here, since it is not presently clear how to do so without building a particular metric into the definition. Instead, in the remainder of the Appendix we mostly study the Carlen--Maas distances~\cite{carlen2020non}, which come equipped with such a noncommutative convex-duality structure and therefore enjoy stronger properties than those forced by the four axioms alone.  Our main result about quantum Lyapunov exponents will apply to any quantum optimal transport distances satisfying our axioms, although of course the Carlen--Maas distances are particularly nice special cases.

\subsection{Upper bound on admissible quantum $p$-Wasserstein distances}

The key object in the proof is the coherent-state noising channel $\mcal N$ defined in Eq.~\eqref{eq:noising-channel}, which we recall here:
\begin{equation}
\label{eq:N-Husimi}
\mcal N[\rho] := \int_{\bbR^{2D}} \Husimi_\rho(\alpha)\ket\alpha\bra\alpha\diff\alpha = \int_{\bbR^{2D}} \langle\alpha|\rho|\alpha\rangle\ket\alpha\bra\alpha\,\frac{\diff\alpha}{(2\pi\hbar)^D}.
\end{equation}
This channel may equivalently be represented as a Gaussian mixture of phase-space translations:
\begin{equation}
\label{eq:N-translations}
\mcal N[\rho] = \int_{\bbR^{2D}} \gamma_\hbar(a)\mcal T_a[\rho]\diff a,
\qquad
\gamma_\hbar(a) := \frac{1}{(2\pi\hbar)^D}\exp\!\left(-\frac{|a|^2}{2\hbar}\right).
\end{equation}
Thus, $\gamma_\hbar$ is the centered Gaussian distribution on $\bbR^{2D}$ with covariance $\hbar\,\Id_{2D}$. The two representations of $\mcal N$ express the same channel: the first measures the state using the coherent-state POVM and prepares the corresponding coherent state, while the second adds a Gaussian random phase-space displacement.

Define
\begin{equation}
M_{D,p} := \left(\int_{\mathbb R^{2D}}|a|^p\gamma_1(a)\diff a\right)^{1/p} = \left(2^{p/2}\frac{\Gamma(D+p/2)}{\Gamma(D)}\right)^{1/p}.
\end{equation}

\begin{proposition}
\label{prp:ub}
Let $p\geq1$, let $d_p$ be an admissible quantum optimal transport distance, and let $\rho,\sigma\in\mathcal D_p$. Then
\begin{equation}
\label{eq:admissible-Husimi-sandwich}
W_p(\Husimi_\rho,\Husimi_\sigma)\leq d_p(\rho,\sigma)\leq W_p(\Husimi_\rho,\Husimi_\sigma)+2M_{D,p}\hbar^{1/2}.
\end{equation}
\end{proposition}

\begin{proof}
The lower bound in Eq.~\eqref{eq:admissible-Husimi-sandwich} is exactly the information-processing axiom, so it remains to prove the upper bound.

We first estimate the distance between a state and its noised version. Using the decompositions
\begin{equation}
\rho = \int_{\bbR^{2D}} \gamma_\hbar(a)\rho\diff a,
\qquad
\mcal N[\rho] = \int_{\bbR^{2D}} \gamma_\hbar(a)\mcal T_a[\rho]\diff a,
\end{equation}
double convexity, applied using the diagonal coupling of $\gamma_\hbar$ with itself, gives
\begin{align}
d_p(\rho,\mcal N[\rho])^p
&\leq \int_{\bbR^{2D}} d_p(\rho,\mcal T_a[\rho])^p\gamma_\hbar(a)\diff a
\nonumber\\
&= \int_{\bbR^{2D}} |a|^p\gamma_\hbar(a)\diff a
\nonumber\\
&= M_{D,p}^p\hbar^{p/2},
\nonumber
\end{align}
where the second line uses the translation-distance axiom. Therefore,
\begin{equation}
\label{eq:noising-close}
d_p(\rho,\mcal N[\rho]) \leq M_{D,p}\hbar^{1/2}.
\end{equation}
The same estimate holds with $\rho$ replaced by $\sigma$.

We next compare the two noised states. Let $\Gamma\in\Pi(\Husimi_\rho,\Husimi_\sigma)$ be any coupling of the Husimi distributions. Applying double convexity to the coherent-state representation~\eqref{eq:N-Husimi} gives
\begin{align}
d_p(\mcal N[\rho],\mcal N[\sigma])^p &\leq \int_{\bbR^{2D}\times\bbR^{2D}}d_p(\ket\alpha\bra\alpha,\ket{\alpha'}\bra{\alpha'})^p\diff\Gamma(\alpha,\alpha'). \nonumber
\end{align}
Since coherent-state projectors are related by phase-space translations,
\begin{equation}
\ket{\alpha'}\bra{\alpha'} = \mcal T_{\alpha'-\alpha}[\ket\alpha\bra\alpha],
\end{equation}
the translation-distance axiom gives
\begin{equation}
d_p(\ket\alpha\bra\alpha,\ket{\alpha'}\bra{\alpha'}) = |\alpha-\alpha'|.
\end{equation}
Consequently,
\begin{equation}
d_p(\mcal N[\rho],\mcal N[\sigma])^p \leq \int_{\bbR^{2D}\times\bbR^{2D}}|\alpha-\alpha'|^p\diff\Gamma(\alpha,\alpha').
\end{equation}
Infimizing over $\Gamma$ yields
\begin{equation}
\label{eq:noised-Wasserstein}
d_p(\mcal N[\rho],\mcal N[\sigma]) \leq W_p(\Husimi_\rho,\Husimi_\sigma).
\end{equation}

Finally, the triangle inequality, together with Eqs.~\eqref{eq:noising-close} and~\eqref{eq:noised-Wasserstein}, gives
\begin{align}
d_p(\rho,\sigma)
&\leq d_p(\rho,\mcal N[\rho]) + d_p(\mcal N[\rho],\mcal N[\sigma]) + d_p(\mcal N[\sigma],\sigma)
\nonumber\\
&\leq W_p(\Husimi_\rho,\Husimi_\sigma) + 2M_{D,p}\hbar^{1/2}.
\nonumber
\end{align}
\end{proof}

The upper bound uses only the metric triangle inequality, double convexity, and the translation-distance axiom. Translation invariance is not needed for this particular argument, while the information-processing axiom gives the complementary lower bound.

\subsection{The $W_1$ quantum optimal transport distance}

The $W_1$ transport distance is constructed by the quantum analogue of the duality against Lipschitz functions:
\[
W_1(\mu_0,\mu_1) = 
\sup_{\|f\|_{\rm Lip} \leq 1}\!\Big\{
\int f \diff \mu_0 - \int f \diff \mu_1\Big\}.
\]
The quantum analogue of a Lipschitz function is a Lipschitz \textit{operator}, which we defined in Section~\ref{sec:Lipsec}.
We recall now the definition of the Lipschitz seminorm of an operator~\eqref{eq:Lipnorm-def}
\begin{equation}
\|A\|_{\rm Lip} := \sup_{|u|=1}\frac{1}{\hbar}\|[u\cdot\hat r,A]\|_{\rm op}.
\end{equation}
The quantum $W_1$ distance is then given by
\begin{equation}
d_{\mathrm{CM},1}(\rho,\sigma) := \sup_{\|A\|_{\rm Lip}\leq1}\Tr\!\left[A(\rho-\sigma)\right],
\end{equation}
where the supremum is understood to be over self-adjoint observables.
That this distance satisfies our axioms is nearly immediate.
\begin{lemma}
The distance $d_{\mathrm{CM},1}$ satisfies the axioms~\ref{ax:tr-invar}-\ref{ax:dp}.
\end{lemma}
\begin{proof}
The translation invariance axiom~\ref{ax:tr-invar} is a consequence of the identity
\[
[u\cdot\hat r,\mathcal T_\alpha^*[A]] = \mathcal T_\alpha^*\!\left([u\cdot\hat r,A]\right),
\]
which immediately implies 
\[
\|\mathcal T_\alpha^*[A]\|_{\rm Lip} = \|A\|_{\rm Lip}.
\]
Axiom~\ref{ax:tr-id} is the identity
\begin{equation}
\label{eq:dist-trans}
d_{\mathrm{CM},1}\bigl(\rho,\mathcal T_\alpha[\rho]\bigr) = |\alpha|.
\end{equation}
First we prove the lower bound. Set $u:=|\alpha|^{-1}\alpha$ and $X_u:=u\cdot\hat r$. Then
\begin{align*}
d_{\mathrm{CM},1}\bigl(\rho,\mathcal T_\alpha[\rho]\bigr)
&\geq \Tr\!\left(-X_u(\rho-\mathcal T_\alpha[\rho])\right) \\
&= \Tr\!\left((-X_u-\mathcal T_\alpha^*[-X_u])\rho\right) \\
&= |\alpha|.
\end{align*}
The upper bound comes from the bound
\begin{align*}
\|A-\mathcal T_\alpha^*[A]\|_{\rm op}
&\leq \int_0^1\|\partial_s(\mathcal T_{s\alpha}^*[A])\|_{\rm op}\diff s \\
&= \int_0^1\frac{1}{\hbar}\|[\hat X_\alpha,A]\|_{\rm op}\diff s \\
&\leq |\alpha|.
\end{align*}

Now we check the convexity axiom~\ref{ax:conv}.  Let $\rho=\int\rho_\theta \mu(\theta)\diff \theta$ and 
$\sigma = \int\sigma_\phi \nu(\phi)\diff \phi$, and let $\Gamma$ be a coupling between $\mu$ and $\nu$.
\begin{align*}
\Tr[A(\rho-\sigma)]
&= \int \Tr[A(\rho_\theta-\sigma_\phi)]\diff\Gamma(\theta,\phi) \\
&\leq \int d_{{\rm CM},1}(\rho_\theta,\sigma_\phi)\diff\Gamma(\theta,\phi).
\end{align*}

Finally, we check the data processing inequality~\ref{ax:dp}.  For classical Lipschitz functions $f\in \Lip(\bbR^{2D})$, we have by 
the coherent-state quantization duality~\eqref{eq:coh-duality}
\[
\int f (\Husimi_\rho - \Husimi_\nu) = 
\Tr[ \Op_\hbar^{\rm coh}(f) (\rho-\nu)] \leq \|\Op_{\hbar}^{\rm coh}(f)\|_{\rm Lip}d_{\mathrm{CM},1}(\rho,\nu).
\]
Combining the operator norm bound for the coherent state quantization~\eqref{eq:coh-bd} with the exact commutator identity~\eqref{eq:coh-comm}, we have $\|\Op_{\hbar}^{\rm coh}(f)\|_{\rm Lip} \leq\|f\|_{\rm Lip}$.   Thus, taking a supremum over Lipschitz $f$ with $\|f\|_{\rm Lip}=1$ proves the desired lower bound.
\end{proof}

To complete the verification of the claimed properties of $d_{\mathrm{CM},1}$ in Lemma 1, we check how $d_{\mathrm{CM},1}$ changes
under linear symplectic transformations
\begin{lemma}
For the channel $\mathcal U_S$ associated with $S\in\operatorname{Sp}(2D,\mathbb R)$,
\begin{equation}
d_{\mathrm{CM},1}\bigl(\mathcal U_S[\rho],\mathcal U_S[\sigma]\bigr)\leq\|S\|d_{\mathrm{CM},1}(\rho,\sigma).
\end{equation}
\end{lemma}
\begin{proof}
This follows from
\begin{equation}
\|\mathcal U_S^*[A]\|_{\rm Lip}\leq\|S\|\|A\|_{\rm Lip}.
\end{equation}
\end{proof}

\section{The $W_2$ quantum optimal transport distance}

\subsection{Preliminaries}

The Carlen--Maas 2-Wasserstein distance~\cite{carlen2020non} is constructed by analogy to the classical Benamou-Brenier formula~\cite{benamou2000computational}:
\begin{equation}
\label{eq:BB-formula}
W_2(\mu_0,\mu_1)^2 = \inf\left\{ \int_0^1\int_{\mathbb R^{2D}} |v_t(\alpha)|^2\mu_t(\alpha)\diff\alpha\diff t : \partial_t\mu_t+\nabla\cdot(\mu_t v_t)=0,\  \mu_{t=0}=\mu_0,\  \mu_{t=1}=\mu_1 \right\}.
\end{equation}
In the formula above, $\mu_t$ is a trajectory along classical densities.  To generalize this formula to the quantum setting, we need appropriate analogues of the following operations: (1) integration over phase space, (2) multiplication by the classical density $\mu$, and (3) differentiation. The basic classical--quantum correspondence is summarized in Table~\ref{tab:classical-quantum-basic}. We will instead work with a dual formulation of the Benamou--Brenier formula, as explained below in Section~\ref{sec:dualaction}, but the primal formulation already motivates the basic ingredients of the quantum construction.

\begin{table*}[t]
\centering
\footnotesize
\setlength{\tabcolsep}{10pt}
\renewcommand{\arraystretch}{1.3}
\begin{tabular}{|p{0.45\textwidth}|p{0.45\textwidth}|}
\hline
\multicolumn{1}{|c|}{\textbf{Classical object}} & \multicolumn{1}{c|}{\textbf{Quantum object}} \\
\hline
Probability density $\mu(\alpha)$ & Density matrix $\rho\in\mathcal D$ \\
\hline
Reference density $\mu_\beta(x)$ & Gibbs state $\sigma_\beta$ \\
\hline
Smooth test function $\phi$ & Self-adjoint observable $A=A^\dagger$ \\
\hline
Bounded function $f$ & Bounded operator $F\in\mathcal B(\mathcal H)$ \\
\hline
Integration $\int f(\alpha)\diff \alpha$ & Trace $\Tr F$ \\
\hline
Expectation $\int f(x)\mu(x)\diff x$ & State--observable pairing $\Tr(F\rho)$ \\
\hline
Signed tangent density $\eta$, with $\int\eta=0$ & Trace-zero tangent $\xi\in\mathcal M_0$ \\
\hline
Dual pairing $\langle\phi,\eta\rangle=\int\phi\,\eta$ & Dual pairing $\Tr(A\xi)$ \\
\hline
Constant function $1$ & Identity operator $\mathds 1$ \\
\hline
\end{tabular}
\caption{Classical and quantum counterparts of the basic objects appearing in the transport formulation.}
\label{tab:classical-quantum-basic}
\end{table*}

Our first task is to clarify the quantum notions of these objects -- multiplication by $\rho$ being the most nontrivial 
given the noncommutativity of the quantum setting.  
For $\omega\in\bbR$ and $\rho\in\mcal{D}$ we define the map
$\Theta_{\rho,\omega}:\mcal{B}_0 \to \mcal{B}$ by
\[
\Theta_{\rho,\omega}(B) := \int_0^1 e^{\omega(1-2s)/2} \rho^{1-s} B \rho^s \diff s\,,
\]
where $\rho^s := \sum_j r_j^s P_j$ is defined spectrally.
We note that
\begin{equation}
\label{eq:Theta-Id}
\Theta_{\rho,\omega}(\Id) =
\int_0^1 e^{\omega(1-2s)/2} \diff s \rho = \kappa_\omega \rho\,,
\end{equation}
with
\begin{equation}
\kappa_\omega := \tfrac2\omega \sinh(\tfrac\omega2)\,,\qquad \kappa_{\beta,\hbar} := \kappa_{2\beta\hbar} = \frac{\sinh(\beta\hbar)}{\beta\hbar},
\end{equation}
with the convention $\kappa_{0,\hbar}=1$.
The identity~\eqref{eq:Theta-Id} reveals that $\Theta_{\rho,\omega}$ is
not in fact a perfect analogy to ``multiplication by $\rho$'' because of
the factor $\kappa_\omega$ which comes from the non-commutativity of $\hat{a}$
and $\hat{a}^\dagger$.  Note however that at fixed $\beta$, $\omega = O(\hbar)$,
so that $\kappa_\omega = 1 + O(\hbar)$ and thus the classical
limit is indeed compatible with the intuition that $\Theta_{\rho,\omega}$
is ``multiplication by $\rho$''.

We note also that $\Theta_{\rho,\omega}$ can be rewritten
\[
\Theta_{\rho,\omega} = \Lambda(e^{\omega/2} L_\rho, e^{-\omega/2}R_\rho),
\]
where $L_\rho[A] = \rho A$ is left multiplication by $\rho$, $R_\rho[A] = A\rho$ is right multiplication
by $\rho$, and $\Lambda$ is the mean
\[
\Lambda(A,B) := \int_0^1 A^{1-s} B^s\diff s.
\]
The point of this formulation is that $\Lambda$ is a Kubo-Ando mean, and is jointly concave in
$A$ and $B$.  Thus, $\Theta_{\rho,\omega}[A]$ is concave in $\rho$ (for fixed $A$).  That is, for fixed $A\in\Lip$, density matrices $\rho_0$ and $\rho_1$,
and $\rho_s := (1-s)\rho_0 + s\rho_1$ for $s\in[0,1]$ we have
\begin{equation}
\label{eq:Theta-concavity}
\Theta_{\rho_s, \omega}(A) \geq (1-s)\Theta_{\rho_0,\omega}(A)
+ s \Theta_{\rho_1,\omega}(A).
\end{equation}

The analogue of the integral for
states is the Hilbert-Schmidt inner product,
\[
\langle A, B\rangle_{\rm HS} := \Tr[A^\dag B],
\]
which is always defined for $A\in\mcal{B}$ and $B\in\mathfrak S_1(\mathcal H)$.
The analogue of a derivative is a commutator with one of the jump operators. For $\ell\in\mathcal I$, define
\begin{equation}
\partial_\ell A := [V_\ell,A].
\end{equation}
The quantum analogue of the classical Dirichlet energy is the quadratic form
\begin{equation}
\label{eq:Q-def}
Q_{\rho,\beta}(A) := \sum_{\ell\in\mathcal I}\left\langle\partial_\ell A,\Theta_{\rho,\omega_\ell}(\partial_\ell A)\right\rangle_{\rm HS}.
\end{equation}
Equivalently, if $X_\ell:=\partial_\ell A$ and $\rho=\sum_n\lambda_nP_n$, then
\begin{equation}
Q_{\rho,\beta}(A) = \sum_{\ell\in\mathcal I}\int_0^1e^{\omega_\ell(1-2s)/2}\sum_{m,n}\lambda_m^{1-s}\lambda_n^s\left|(X_\ell)_{mn}\right|^2\diff s.
\end{equation}
In particular, $Q_{\rho,\beta}(A)\geq0$. Its polarization is
\begin{equation}
\Gamma_{\rho,\beta}(A,B) := \sum_{\ell\in\mathcal I}\left\langle\partial_\ell A,\Theta_{\rho,\omega_\ell}(\partial_\ell B)\right\rangle_{\rm HS},
\end{equation}
so that $Q_{\rho,\beta}(A)=\Gamma_{\rho,\beta}(A,A)$. The bilinearity of $\Gamma_{\rho,\beta}(A,B)$ and the nonnegativity of $\Gamma_{\rho,\beta}(A,A)$ then imply the following Cauchy-Schwarz inequality:
\begin{equation}
\label{eq:Gamma-CS}
|\Gamma_{\rho,\beta}(A,B)| \leq \sqrt{Q_{\rho,\beta}(A) Q_{\rho,\beta}(B)}\,.
\end{equation}

The quadratic form $Q_{\rho,\beta}(A)$ is the quantum analogue of the energy $\int |\nabla \phi|^2 \rho$.  In this analogy, the observable $A$ plays the role of the classical function $\phi$, while the gradient is represented by the commutators $[V_\ell,A]$.

The concavity of $\Theta$ implies the following concavity result for
the quadratic form $Q$.
\begin{lemma}
\label{lem:Q-convexity}
For fixed $A\in\mathcal B_0$, the map $\rho\mapsto Q_{\rho,\beta}(A)$ is concave on $\mathcal D_+$. In particular, if $\rho_\theta=(1-\theta)\rho_0+\theta\rho_1$, then
\begin{equation}
Q_{\rho_\theta,\beta}(A) \geq (1-\theta)\,Q_{\rho_0,\beta}(A)+\theta \, Q_{\rho_1,\beta}(A).
\end{equation}
\end{lemma}

The Dirichlet energy $Q_{\rho,\beta}$ has an additional translation-invariance property that will imply the translation invariance axiom of the distance.

\begin{lemma}
\label{lem:Q-translation}
The quadratic form $Q_{\rho,\beta}[A]$ satisfies the translation invariance
\[
Q_{\mathcal T_\alpha[\rho],\beta}[\mathcal T_\alpha[A]] = Q_{\rho,\beta}[A].
\]
\end{lemma}
\begin{proof}
This follows from the definition of $Q_{\rho,\beta}$ and the identity
\[
[V_\ell,\mathcal T_\alpha[A]] = \mathcal T_\alpha[V_\ell,A]\,.
\]
\end{proof}

We also record how $Q_{\rho,\beta}$ transforms under linear symplectomorphisms. 
\begin{lemma}
\label{lem:Q-symplecto}
Let $S\in\operatorname{Sp}(2D,\mathbb R)$. Then
\begin{equation}
Q_{\mathcal U_S[\rho],0}\bigl(\mathcal U_S[A]\bigr)\leq\|S\|^2Q_{\rho,0}(A).
\end{equation}
\end{lemma}
\begin{proof}
First, we use
\begin{equation}
\Theta_{\mathcal U_S[\rho],0}\bigl(\mathcal U_S[X]\bigr) = \mathcal U_S\!\left[\Theta_{\rho,0}(X)\right],
\end{equation}
and
\begin{equation}
[\hat r,\mathcal U_S[A]] = \mathcal U_S\!\left([S\hat r,A]\right).
\end{equation}
Set $X_k=[\hat r_k,A]$. By the definition of $\Theta_{\rho,0}$,
\begin{align*}
Q_{\mathcal U_S[\rho],0}(\mathcal U_S[A])
&=
\int_0^1 \sum_j
\left\|
\rho^{(1-s)/2}
\left(\sum_k S_{jk}X_k\right)
\rho^{s/2}
\right\|_{\mathrm{HS}}^2\diff s \\
&=
\int_0^1 \sum_j
\left\|
\sum_k S_{jk}
\rho^{(1-s)/2}X_k\rho^{s/2}
\right\|_{\mathrm{HS}}^2\diff s \\
&\leq
\|S\|^2
\int_0^1 \sum_k
\left\|
\rho^{(1-s)/2}X_k\rho^{s/2}
\right\|_{\mathrm{HS}}^2\diff s \\
&=
\|S\|^2Q_{\rho,0}(A).
\end{align*}
\end{proof}

\subsection{The quantum optimal transport distance}
\label{sec:dualaction}

\begin{table*}[t]
\centering
\footnotesize
\setlength{\tabcolsep}{10pt}
\renewcommand{\arraystretch}{1.3}
\begin{tabular}{|p{0.45\textwidth}|p{0.45\textwidth}|}
\hline
\multicolumn{1}{|c|}{\textbf{Classical transport object}} & \multicolumn{1}{c|}{\textbf{Quantum optimal transport object}} \\
\hline
Gradient $\nabla\phi$ & Noncommutative gradients $\partial_\ell A:=[V_\ell,A]$ \\
\hline
Product $f\mu$ & Operator mean $\Theta_{\rho,\omega_\ell}(F)$ \\
\hline
Weighted gradient $\mu\nabla\phi$ & Noncommutative flux $\Theta_{\rho,\omega_\ell}(\partial_\ell A)$ \\
\hline
Cotangent quadratic form $Q_\mu^{\rm cl}(\phi):=\int|\nabla\phi|^2\mu\,\diff x$ & Cotangent quadratic form \newline $Q_{\rho,\beta}(A):=\sum_\ell\Tr\!\left[(\partial_\ell A)^\dagger\Theta_{\rho,\omega_\ell}(\partial_\ell A)\right]$ \\
\hline
Dual tangent action $\mathsf a_\mu^{\rm cl}(\eta):=\sup_\phi\{2\langle\phi,\eta\rangle-Q_\mu^{\rm cl}(\phi)\}$ & Dual tangent action $\mathsf a_{\rho,\beta}(\xi):=\sup_A\{2\Tr(A\xi)-Q_{\rho,\beta}(A)\}$ \\
\hline
Benamou--Brenier action $\int_0^1\mathsf a_{\mu_t}^{\rm cl}(\dot\mu_t)\diff t$ & Quantum action $\int_0^1\mathsf a_{\rho_t,\beta}(\dot\rho_t)\diff t$ \\
\hline
\end{tabular}
\caption{Classical and quantum counterparts of the differential and variational objects defining the transport geometry.}
\label{tab:classical-quantum-transport}
\end{table*}

Having introduced the ingredients of the quantum transport geometry, we summarize their classical counterparts in Table~\ref{tab:classical-quantum-transport}. We can now formulate the quantum analogue of the Benamou--Brenier construction~\eqref{eq:BB-formula}.  In the infinite-dimensional setting, however, the primal formulation of~\eqref{eq:BB-formula} is difficult to make precise.  The difficulty is that the usual way to use the primal formulation is to construct test trajectories $\rho_s$, and recover the velocity field $v_s$ as the solution to the equation
\[ \nabla \cdot (v_s \rho_s) = -\partial_s \rho_s. \]
In the infinite-dimensional quantum setting, the corresponding non-commutative divergence equation can be highly degenerate and difficult to formulate, let alone solve. Instead, it is more natural to work with the dual formulation. To dualize, define the classical primal tangent action by
\begin{equation}
\label{eq:primal-a-classical}
\mathsf a_{\mu}^{\rm cl,primal}(\eta) := \inf\left\{ \int_{\mathbb R^{2D}}|v(\alpha)|^2\mu(\alpha)\diff\alpha : \eta=-\nabla\cdot(\mu v) \right\}.
\end{equation}
The corresponding dual action is
\begin{equation}
\label{eq:dual-a-classical}
\mathsf a_{\mu}^{\rm cl}(\eta) := \sup_\phi \left\{ 2\langle\phi,\eta\rangle - \int_{\mathbb R^{2D}}|\nabla\phi(\alpha)|^2\mu(\alpha)\diff\alpha
\right\}.
\end{equation}
Integrating by parts and using $2a\cdot b-|a|^2\leq|b|^2$ gives
\begin{equation}
\mathsf a_{\mu}^{\rm cl}(\eta) \leq \mathsf a_{\mu}^{\rm cl,primal}(\eta).
\end{equation}
Conversely, when $\eta=-\nabla\cdot(\mu\nabla\psi)$, choosing $\phi=\psi$ saturates the inequality. Thus we have
\begin{equation}
\label{eq:dual-classical-def}
W_2(\mu_0,\mu_1)^2 = \inf_{\mu_t} \int_0^1 \mathsf a_{\mu_t}^{\rm cl}(\dot\mu_t) \diff t.
\end{equation}

The characterization of the classical $W_2$ distance in terms of the dual action~\eqref{eq:dual-a-classical} is more natural to adapt to the quantum setting because it uses only the structures defined above and avoids solving the divergence equation appearing in the primal formulation~\eqref{eq:primal-a-classical}.

We now proceed with the definition of the quantum optimal transport distance in stages. The first quantity to define is the quantum analogue of the dual action $\msf{a}_\rho$.  We note a few small points: (1) we define the action for tangent vectors which should be trace-class and traceless, (2) the dual action takes as input the additional parameter $\beta$ which is not unique to the quantum setting, and (3) we take the supremum over only Lipschitz observables so that the quadratic form $Q_{\rho,\beta}$
is well defined.
\begin{definition}[Quantum dual action]
For tangent vectors $\xi\in\mcal{M}_0$, we define the quantum action
\begin{equation}
\label{eq:quantum-dual-action}
\mathsf a_{\rho,\beta}(\xi) := \sup_{A\in\mathcal B_0}\left\{2\Tr(A\xi)-Q_{\rho,\beta}(A)\right\}.
\end{equation}
\end{definition}

\noindent
We observe the following transformation rules for the action under translations and linear
symplectomorphisms.

\begin{lemma}
The action satisfies
\begin{equation}
\label{eq:action-invariance}
\mathsf a_{\mathcal T_\alpha[\rho],\beta}\bigl(\mathcal T_\alpha[\xi]\bigr) = \mathsf a_{\rho,\beta}(\xi).
\end{equation}
Moreover, at $\beta=0$,
\begin{equation}
\label{eq:action-symplectic}
\|S\|^{-2}\mathsf a_{\rho,0}(\xi) \leq \mathsf a_{\mathcal U_S[\rho],0}\bigl(\mathcal U_S[\xi]\bigr) \leq \|S\|^2\mathsf a_{\rho,0}(\xi).
\end{equation}
\end{lemma}
\begin{proof}
The first follows directly from Lemma~\ref{lem:Q-translation}, using also that $\mathcal T_\alpha$ is unitary with respect to the Hilbert-Schmidt inner product. For the symplectic transformation, use the change of variables $B=\mathcal U_S[A]$ in the variational formula: \begin{align*} \mathsf a_{\mathcal U_S[\rho],0}\bigl(\mathcal U_S[\xi]\bigr) &= \sup_{A\in\mathcal B_0}\left\{2\Tr\!\left(\mathcal U_S[A]\mathcal U_S[\xi]\right)-Q_{\mathcal U_S[\rho],0}\bigl(\mathcal U_S[A]\bigr)\right\}\\ &\geq \sup_{A\in\mathcal B_0}\left\{2\Tr(A\xi)-\|S\|^2Q_{\rho,0}(A)\right\}\\ &= \|S\|^{-2}\mathsf a_{\rho,0}(\xi). \end{align*} The upper bound follows by applying the lower bound to $S^{-1}$ and using $\|S^{-1}\|=\|S\|$.
\end{proof}

To define the quantum optimal transport distance we need to resolve one additional subtlety in the definition of~\eqref{eq:dual-classical-def}, which is that we want to take an infimum over \textit{continuous} curves $\rho_t$. We do this by simply taking the lower-semicontinuous envelope over the \textit{differentiable} curves.
\begin{definition}[Quantum optimal transport distance]
For a differentiable curve $\rho_t\in C^1([0,1]; \mcal{D})$, define
\[
\msf{A}_\beta^0(\rho) := \int_0^1 \msf{a}_{\rho_t,\beta}(\dot{\rho}_t)\diff t.
\]
Then define $\msf{A}_\beta$ to be the lower-semicontinuous envelope of $\msf{A}_\beta^0$ under trace-norm convergence of curves $\rho_t$.  That is, for $\rho\in C^0([0,1];\mcal{D})$ we define
\[
\msf{A}_\beta(\rho) := \inf\Big\{ \liminf_{n\to\infty} \msf{A}_\beta^0(\tilde{\rho}_n) \mid \tilde{\rho}_n\in C^1([0,1];\mcal{D}) \text{ and } \lim_{n\to\infty} \|\tilde{\rho}_n(s)-\rho(s)\|_{\Tr} \to 0 \text{ for all }s\in [0,1]\Big\}.
\]
The quantum optimal transport distance is then given by
\begin{align}
\label{eq:CM-raw-def}
\widetilde d_{\beta,2}(\rho_0,\rho_1)^2 &:= \inf\left\{\mathsf A_\beta(\rho):\rho\in C^0([0,1];\mathcal D),\ \rho(0)=\rho_0,\ \rho(1)=\rho_1\right\} \\
\label{eq:CM-def}
d_{\mathrm{CM},2}^{(\beta)}(\rho_0,\rho_1) &:= \sqrt{\kappa_{\beta,\hbar}}\,\widetilde d_{\beta,2}(\rho_0,\rho_1)\,.
\end{align}
\end{definition}
\noindent At $\beta=0$, we abbreviate
\begin{equation}
d_{\mathrm{CM},2} := d_{\mathrm{CM},2}^{(0)}.
\end{equation}

Now we state an important convexity result for the distance.  The proof follows the same basic argument as the analogous result, Theorem 9.7 of Carlen and Maas~\cite{carlen2020non}.  Indeed, our task is made much easier because we have defined the distance in the dual formulation, which is simpler to work with (whereas the difficult part of Theorem 9.7 is to establish that the dual formulation is equivalent to the primal formulation).
\begin{theorem}
\label{thm:distance-convexity}
Let $\theta\in(0,1)$, and $\rho_j,\nu_j\in\mcal{D}$ for $j\in\{0,1\}$.  Set
$\rho_\theta := (1-\theta)\rho_0 + \theta \rho_1$ and likewise define $\nu_\theta$.  Then
\begin{equation}
\label{eq:distance-convexity}
\widetilde d_{\beta,2}(\rho_\theta,\nu_\theta)^2 \leq (1-\theta)\,\widetilde d_{\beta,2}(\rho_0,\nu_0)^2
+ \theta \, \widetilde d_{\beta,2}(\rho_1,\nu_1)^2.
\end{equation}
In general, for mixtures of $n$ states of the form $\tilde{\rho} = \sum_{j=1}^n \theta_j\rho_j$ and
$\tilde{\nu} = \sum_{j=1}^n \theta_j \nu_j$
with $0\leq \theta_j$ and $\sum \theta_j=1$,
we have
\begin{equation}
\label{eq:multipoint-convexity}
\widetilde d_{\beta,2}(\tilde{\rho},\tilde{\nu})^2 \leq \sum_{j=1}^n \theta_j\widetilde d_{\beta,2}(\rho_j,\nu_j)^2.
\end{equation}
\end{theorem}
\begin{proof}
This follows from the following statement of convexity of the action $\msf{a}_{\rho,\beta}(\xi)$: for
$\xi_\theta := (1-\theta)\xi_0 + \theta\xi_1$ and $\rho_\theta$ as above, we have
\begin{equation}
\label{eq:action-convexity}
\msf{a}_{\rho_\theta,\beta}(\xi_\theta) \leq (1-\theta)\,\msf{a}_{\rho_0,\beta}(\xi_0) + \theta \, \msf{a}_{\rho_1,\beta}(\xi_1).
\end{equation}
This in turn follows from the concavity of $Q_{\rho,\beta}(A)$ in $\rho$. To deduce~\eqref{eq:distance-convexity}, apply~\eqref{eq:action-convexity} to the convex combination $(1-\theta)\gamma_0 + \theta\gamma_1$ of two almost-minimizing curves $\gamma_j$ connecting $\rho_j$ to $\nu_j$. The bound~\eqref{eq:multipoint-convexity} follows from~\eqref{eq:distance-convexity} by induction on $n$.
\end{proof}

Several of the axioms of quantum optimal transport distance can  now be verified for the distance $\widetilde d_{\beta,2}$ by appropriately integrating properties of the action.
\begin{corollary}
For every $\beta$, the normalized distance $d_{\mathrm{CM},2}^{(\beta)}$ satisfies translation invariance and double convexity. At $\beta=0$, the distance $d_{\mathrm{CM},2}:=d_{\mathrm{CM},2}^{(0)}$ also satisfies
\begin{equation}
\|S\|^{-1}d_{\mathrm{CM},2}(\rho,\nu) \leq d_{\mathrm{CM},2}\bigl(\mathcal U_S[\rho],\mathcal U_S[\nu]\bigr) \leq \|S\|d_{\mathrm{CM},2}(\rho,\nu).
\end{equation}
\end{corollary}
\begin{proof}
Axiom~\ref{ax:tr-invar} follows from integrating~\eqref{eq:action-invariance}.  
The convexity axiom~\ref{ax:conv} follows from Theorem~\ref{thm:distance-convexity} after approximating the integrals by Riemann sums. Finally, the bi-Lipschitz property of symplectic maps follows from integrating~\eqref{eq:action-symplectic}.
\end{proof}

The remaining axioms~\ref{ax:tr-id} and~\ref{ax:dp} require more involved verifications, each of which occupies its own section. After these are established, the verification of admissibility of $d_{\mathrm{CM},2}^{(\beta)}$ is complete.

\section{Carlen--Maas quantum optimal transport distance for translations}

In this section we verify the translation identity of Axiom~\ref{ax:tr-id}. We begin with two identities for linear observables. For $\alpha = (q,p) \in \mathbb R^{2D}$, define
\begin{equation}
L_\alpha := \alpha \cdot \hat r = q \cdot \hat x + p \cdot \hat p.
\end{equation}
Recall that $\hat X_\alpha = \alpha^\intercal J\hat r = p \cdot \hat x - q \cdot \hat p$ generates the translation channel $\mathcal T_{t\alpha}$.
\begin{lemma}
\label{lem:linear-translation-potential}
For every state $\rho$ and every admissible observable $A$,
\begin{equation}
Q_{\rho,\beta}(L_\alpha) = \kappa_{\beta,\hbar}|\alpha|^2
\end{equation}
and
\begin{equation}
\Gamma_{\rho,\beta}(A,L_\alpha) = \kappa_{\beta,\hbar}\frac{i}{\hbar}\Tr\!\left(A[\hat X_\alpha,\rho]\right).
\end{equation}
\end{lemma}
\begin{proof}
For each $j = 1,\ldots,D$, a direct computation gives
\begin{equation}
[V_{j,-},L_\alpha] = \frac{q_j + ip_j}{\sqrt{2}}\,\mathds 1,
\qquad
[V_{j,+},L_\alpha] = \frac{-q_j + ip_j}{\sqrt{2}}\,\mathds 1.
\end{equation}
Since $\kappa_{\omega_{j,\pm}} = \kappa_{\beta,\hbar}$ and $\Theta_{\rho,\omega}(\mathds 1) = \kappa_\omega\rho$, it follows that
\begin{align}
Q_{\rho,\beta}(L_\alpha)
&= \kappa_{\beta,\hbar}\sum_{j = 1}^D\left(\frac{|q_j + ip_j|^2}{2} + \frac{|-q_j + ip_j|^2}{2}\right)\\
&= \kappa_{\beta,\hbar}|\alpha|^2.
\end{align}
Writing $c_\ell := [V_\ell,L_\alpha]$, we also have
\begin{equation}
\sum_{\ell \in \mathcal I}c_\ell V_\ell^\dagger = \frac{i}{\hbar}\hat X_\alpha.
\end{equation}
Therefore,
\begin{align}
\Gamma_{\rho,\beta}(A,L_\alpha)
&= \kappa_{\beta,\hbar}\sum_{\ell \in \mathcal I}c_\ell\Tr\!\left([V_\ell,A]^\dagger\rho\right)\\
&= \kappa_{\beta,\hbar}\Tr\!\left(A\left[\sum_{\ell \in \mathcal I}c_\ell V_\ell^\dagger,\rho\right]\right)\\
&= \kappa_{\beta,\hbar}\frac{i}{\hbar}\Tr\!\left(A[\hat X_\alpha,\rho]\right).
\end{align}
\end{proof}
To prove the translation-distance identity, we first establish an upper bound for the distance between a state and its translate.
\begin{proposition}
\label{prp:translation-ub}
For every density matrix $\rho \in \mathcal D$ and every $\alpha \in \mathbb R^{2D}$,
\begin{equation}
\widetilde d_{\beta,2}\bigl(\rho,\mathcal T_\alpha[\rho]\bigr) \leq \kappa_{\beta,\hbar}^{-1/2}|\alpha|.
\end{equation}
\end{proposition}
\begin{proof}
For $N \in \mathbb N$, let
\begin{equation}
\Pi_N := \sum_{|\mathbf n| \leq N}|\mathbf n\rangle\langle\mathbf n|,
\qquad
c_N := \Tr(\Pi_N\rho),
\qquad
\rho_N := c_N^{-1}\Pi_N\rho\Pi_N,
\end{equation}
where $N$ is sufficiently large that $c_N > 0$. Consider the translated curve
\begin{equation}
\rho_N(t) := \mathcal T_{t\alpha}[\rho_N],
\qquad
0 \leq t \leq 1.
\end{equation}
This is a differentiable finite-rank curve satisfying
\begin{equation}
\dot\rho_N(t) = \frac{i}{\hbar}[\hat X_\alpha,\rho_N(t)].
\end{equation}
Lemma~\ref{lem:linear-translation-potential} therefore gives
\begin{equation}
\Tr\!\left(A\dot\rho_N(t)\right) = \kappa_{\beta,\hbar}^{-1}\Gamma_{\rho_N(t),\beta}(A,L_\alpha).
\end{equation}
Using the Cauchy--Schwarz inequality for $\Gamma_{\rho,\beta}$, we obtain
\begin{align}
2\Tr\!\left(A\dot\rho_N(t)\right) - Q_{\rho_N(t),\beta}(A)
&\leq 2\kappa_{\beta,\hbar}^{-1}\left|\Gamma_{\rho_N(t),\beta}(A,L_\alpha)\right| - Q_{\rho_N(t),\beta}(A)\\
&\leq 2\kappa_{\beta,\hbar}^{-1}Q_{\rho_N(t),\beta}(A)^{1/2}Q_{\rho_N(t),\beta}(L_\alpha)^{1/2} - Q_{\rho_N(t),\beta}(A)\\
&\leq \kappa_{\beta,\hbar}^{-2}Q_{\rho_N(t),\beta}(L_\alpha)\\
&= \kappa_{\beta,\hbar}^{-1}|\alpha|^2.
\end{align}
Taking the supremum over $A \in \mathcal B_0$ gives
\begin{equation}
\mathsf a_{\rho_N(t),\beta}\bigl(\dot\rho_N(t)\bigr) \leq \kappa_{\beta,\hbar}^{-1}|\alpha|^2,
\end{equation}
and hence
\begin{equation}
\mathsf A_\beta^0(\rho_N) \leq \kappa_{\beta,\hbar}^{-1}|\alpha|^2.
\end{equation}
Since $\rho_N \to \rho$ in trace norm and
\begin{equation}
\sup_{0 \leq t \leq 1}\left\|\mathcal T_{t\alpha}[\rho_N] - \mathcal T_{t\alpha}[\rho]\right\|_1 = \|\rho_N - \rho\|_1 \longrightarrow 0,
\end{equation}
by the definition of $\mathsf A_\beta$ as the lower-semicontinuous envelope,
\begin{equation}
\mathsf A_\beta\bigl(t \mapsto \mathcal T_{t\alpha}[\rho]\bigr) \leq \kappa_{\beta,\hbar}^{-1}|\alpha|^2.
\end{equation}
Taking the infimum over curves proves the result.
\end{proof}
The matching lower bound follows by testing the dual action against bounded regularizations of $L_\alpha$.
\begin{proposition}
\label{prp:translation-lb}
For every density matrix $\rho \in \mathcal D$ and every $\alpha \in \mathbb R^{2D}$,
\begin{equation}
\widetilde d_{\beta,2}\bigl(\rho,\mathcal T_\alpha[\rho]\bigr) = \kappa_{\beta,\hbar}^{-1/2}|\alpha|.
\end{equation}
\end{proposition}
\begin{proof}
The case $\alpha = 0$ is immediate, so assume $\alpha \neq 0$. For $R > 0$, define the bounded self-adjoint observable
\begin{equation}
F_R := f_R(L_\alpha),
\qquad
f_R(s) := R\arctan(s/R).
\end{equation}
Since $[V_\ell,L_\alpha]$ is a scalar multiple of the identity,
\begin{equation}
[V_\ell,F_R] = [V_\ell,L_\alpha]f_R'(L_\alpha).
\end{equation}
Moreover, $\|f_R'\|_\infty \leq 1$. The spectral representation of $\Theta_{\rho,\omega}$ and the weighted arithmetic--geometric mean inequality imply
\begin{equation}
\left\langle B,\Theta_{\rho,\omega}(B)\right\rangle_{\rm HS} \leq \kappa_\omega\|B\|_{\rm op}^2
\end{equation}
for every bounded operator $B$. Consequently,
\begin{equation}
Q_{\rho,\beta}(F_R) \leq \kappa_{\beta,\hbar}|\alpha|^2
\end{equation}
for every state $\rho$. Let $\gamma \in C^1([0,1];\mathcal D)$ be any differentiable curve. Testing the dual action against scalar multiples $sF_R$, with $s \in \mathbb R$, gives
\begin{align}
\mathsf a_{\gamma(t),\beta}\bigl(\dot\gamma(t)\bigr)
&\geq \sup_{s \in \mathbb R}\left\{2s\Tr\!\left(F_R\dot\gamma(t)\right) - s^2\kappa_{\beta,\hbar}|\alpha|^2\right\}\\
&= \frac{\left|\Tr\!\left(F_R\dot\gamma(t)\right)\right|^2}{\kappa_{\beta,\hbar}|\alpha|^2}.
\end{align}
Integrating in time and applying the Cauchy--Schwarz inequality gives
\begin{equation}
\mathsf A_\beta^0(\gamma) \geq \frac{\left|\Tr\!\left(F_R(\gamma(1) - \gamma(0))\right)\right|^2}{\kappa_{\beta,\hbar}|\alpha|^2}.
\end{equation}
Because $F_R$ is bounded, this estimate is stable under trace-norm convergence and therefore passes to the lower-semicontinuous envelope $\mathsf A_\beta$. Now suppose that $\gamma(0) = \rho$ and $\gamma(1) = \mathcal T_\alpha[\rho]$. Since
\begin{equation}
\mathcal T_\alpha^*[L_\alpha] = L_\alpha + |\alpha|^2\mathds 1,
\end{equation}
functional calculus gives
\begin{equation}
\mathcal T_\alpha^*[F_R] = f_R\!\left(L_\alpha + |\alpha|^2\mathds 1\right).
\end{equation}
Therefore,
\begin{equation}
\Tr\!\left(F_R(\mathcal T_\alpha[\rho] - \rho)\right) = \Tr\!\left(\left[f_R\!\left(L_\alpha + |\alpha|^2\mathds 1\right) - f_R(L_\alpha)\right]\rho\right).
\end{equation}
As $R \to \infty$,
\begin{equation}
f_R\!\left(L_\alpha + |\alpha|^2\mathds 1\right) - f_R(L_\alpha) \longrightarrow |\alpha|^2\mathds 1
\end{equation}
strongly, while its operator norm is bounded by $|\alpha|^2$. Hence
\begin{equation}
\lim_{R \to \infty}\Tr\!\left(F_R(\mathcal T_\alpha[\rho] - \rho)\right) = |\alpha|^2.
\end{equation}
It follows that every curve joining $\rho$ to $\mathcal T_\alpha[\rho]$ satisfies
\begin{equation}
\mathsf A_\beta(\gamma) \geq \kappa_{\beta,\hbar}^{-1}|\alpha|^2.
\end{equation}
Taking the infimum over curves and combining with Proposition~\ref{prp:translation-ub} proves the equality.
\end{proof}
\begin{corollary}
The normalized Carlen--Maas distance satisfies
\begin{equation}
d_{\mathrm{CM},2}^{(\beta)}\bigl(\rho,\mathcal T_\alpha[\rho]\bigr) = |\alpha|.
\end{equation}
\end{corollary}

\section{The data processing inequality by Petz monotonicity}
\label{sec:petz}

What remains is to verify the data processing axiom~\ref{ax:dp} for $d_{\mathrm{CM},2}$.
The Husimi measurement map
\begin{equation}
\mathcal C:\mathfrak S_1(\mathcal H)\to L^1(\mathbb R^{2D})\,, \qquad \mathcal C[\rho] = \Husimi_\rho\,.
\end{equation}
is positive and trace preserving. Its adjoint
\begin{equation}
\mathcal C^*:L^\infty(\mathbb R^{2D})\to\mathcal B(\mathcal H)\,,
\end{equation}
defined in Eq.~\eqref{eq:anti-Wick}, is normal, unital, and completely positive. Recall that it is given by anti-Wick quantization,
\begin{equation}
\mathcal C^*[f] = \int_{\mathbb R^{2D}}f(\alpha)|\alpha\rangle\langle\alpha|\frac{\diff\alpha}{(2\pi\hbar)^D}.
\end{equation}
The main result, from which the lower bound for $\widetilde d_{\beta,2}$ follows, is the following \textit{upper} bound for the quadratic form $Q_{\rho,\beta}(A)$.
\begin{proposition}
\label{prp:Q-ub}
For any $g\in C_c^\infty(\bbR^{2D})$, the quantum Dirichlet energy satisfies
\[
Q_{\rho,\beta}(\mcal{C}^*g)
\leq \kappa_{\beta,\hbar} \int_{\bbR^{2D}} |\nabla g(\alpha)|^2 \Husimi_\rho(\alpha) \diff \alpha\,.
\]
\end{proposition}

The proof is based on Petz monotonicity.  Petz monotonicity is a deep monotonicity statement concerning completely positive maps and certain quadratic forms known as quasi-entropies.  The application to our setting uses the normal, unital, completely positive map $\mathcal C^*:L^\infty(\mathbb R^{2D})\to\mathcal B(\mathcal H)$, whose domain is a \textit{commutative}
von Neumann algebra with multiplication.

The special case of Petz's monotonicity that we need is stated below.

\begin{lemma}[Corollary of Petz monotonicity~\cite{petz1985quasientropies}]
\label{lem:Petz-cor}
Let $\rho\in\mcal{D}$, $\omega\in\bbR$, and
$f\in L^\infty(\bbR^{2D})$. Then
\begin{equation}
\label{eq:Petz-consequence}
\left\langle \mcal{C}^*f,\,
\Theta_{\rho,\omega}(\mcal{C}^*f)\right\rangle_{\rm HS}
\leq
\kappa_\omega
\int_{\bbR^{2D}} |f(\alpha)|^2
\Husimi_\rho(\alpha)\diff\alpha .
\end{equation}
\end{lemma}
\noindent Before we explain how Lemma~\ref{lem:Petz-cor} follows from Petz monotonicity,
we first use it to derive Proposition~\ref{prp:Q-ub}.

\begin{proof}[Proof of Proposition~\ref{prp:Q-ub} using Lemma~\ref{lem:Petz-cor}]
Unwrapping the definition of $Q_{\rho,\beta}$ in~\eqref{eq:Q-def}, we have
\begin{equation}
Q_{\rho,\beta}(\mathcal C^*[g]) = \sum_{\ell\in\mathcal I}\left\langle[V_\ell,\mathcal C^*[g]],\Theta_{\rho,\omega_\ell}\bigl([V_\ell,\mathcal C^*[g]]\bigr)\right\rangle_{\rm HS}.
\end{equation}
The result then follows by applying Lemma~\ref{lem:Petz-cor} to each commutator and using the exact commutator identity~\eqref{eq:coh-comm} for the linear combinations defining $V_{j,\pm}$. Since $\kappa_{-\omega} = \kappa_\omega$, every term carries the same factor $\kappa_{\beta,\hbar}$.
\end{proof}

Before we complete the proof of Lemma~\ref{lem:Petz-cor}, we point out that Proposition~\ref{prp:Q-ub} implies a lower bound for the quantum transport distance.
\begin{proof}[Proof that CM satisfies Axiom~\ref{ax:dp}]
The upper bound on the quadratic form $Q_{\rho,\beta}$ implies the following lower bound on the action in terms of the classical action.
\begin{align}
\mathsf a_{\rho,\beta}(\xi)
&\geq \sup_{f\in C^1(\mathbb R^{2D})}\left\{2\int f(\alpha)\Husimi_\xi(\alpha)\diff\alpha-\kappa_{\beta,\hbar}\int|\nabla f(\alpha)|^2\Husimi_\rho(\alpha)\diff\alpha\right\}\\
&=\kappa_{\beta,\hbar}^{-1}\mathsf a_{\Husimi_\rho}^{\rm cl}(\Husimi_\xi).
\end{align}
Integrating over infimizing curves $\rho_t$ and using the characterization~\eqref{eq:dual-classical-def} for the classical transport distance gives
\begin{equation}
\widetilde d_{\beta,2}(\rho,\sigma)\geq\kappa_{\beta,\hbar}^{-1/2}W_2(\Husimi_\rho,\Husimi_\sigma).
\end{equation}
Consequently,
\begin{equation}
W_2(\Husimi_\rho,\Husimi_\sigma)\leq d_{\mathrm{CM},2}^{(\beta)}(\rho,\sigma).
\end{equation}
\end{proof}

\begin{proof}[Proof of Lemma~\ref{lem:Petz-cor} using Petz monotonicity]
Let
\[
\ell(t):=\int_0^1 t^s\diff s
=
\begin{cases}
\dfrac{t-1}{\log t},&t\neq1,\\[4pt]
1,&t=1.
\end{cases}
\]
The function $\ell$ is operator monotone on $[0,\infty)$, since
$t\mapsto t^s$ is operator monotone for every $s\in[0,1]$. Note that the anti-Wick quantization map $\mcal{C}^*:L^\infty(\bbR^{2D})\to \mcal{B}(\mcal{H})$ is normal, unital, and completely positive, hence in particular it is $2$-positive.

We first suppose that $\rho$ is faithful. Define positive normal
functionals on $\mcal{B}(\mcal{H})$ by
\[
\varphi_\pm(A):= e^{\pm\omega/2}\Tr(\rho A)\,,
\]
and positive normal functionals on $L^\infty(\bbR^{2D})$ by
\[
\varphi_{\pm,0}(g):= e^{\pm\omega/2} \int_{\bbR^{2D}} g(\alpha)\Husimi_\rho(\alpha)\diff\alpha\,.
\]
The duality between the Husimi measurement and anti-Wick quantization,
\begin{equation}
\label{eq:Petz-domination}
\varphi_{\pm,0}=\varphi_\pm\circ\mcal{C}^*
\end{equation}
verifies the hypotheses\footnote{In the notation of Petz's result, $M_0=L^\infty(\bbR^{2D})$, $M=\mcal{B}(\mcal{H})$, $\alpha=\mcal{C}^*$, $\varphi=\varphi_+$, $\varphi_0=\varphi_{+,0}$, $\omega=\varphi_-$, and $\omega_0=\varphi_{-,0}$.} $\varphi\circ\alpha\leq\varphi_0$ and $\omega\circ\alpha\leq\omega_0$ of Petz's theorem, in fact with equality. Therefore Petz's monotonicity theorem~\cite[Theorem 4]{petz1985quasientropies} applies and yields the bound
\begin{equation}
\label{eq:Petz-abstract}
S_\ell^{\mcal C^*f}(\varphi_+,\varphi_-) \leq S_\ell^f(\varphi_{+,0},\varphi_{-,0}).
\end{equation}

For faithful $\rho$, the standard-form expression for the quasi-entropy gives, for every bounded operator $X$,
\begin{align*}
S_\ell^X(\varphi_+,\varphi_-)
&=
\int_0^1 e^{\omega(s-1/2)}
\Tr\!\left(X^\dagger\rho^sX\rho^{1-s}\right)\diff s\\
&=
\int_0^1 e^{\omega(1-2s)/2}
\Tr\!\left(X^\dagger\rho^{1-s}X\rho^s\right)\diff s\\
&=
\left\langle X,\Theta_{\rho,\omega}(X)
\right\rangle_{\rm HS}.
\end{align*}
Thus
\[
S_\ell^{\mcal{C}^*f}(\varphi_+,\varphi_-)
=
\left\langle
\mcal{C}^*f,
\Theta_{\rho,\omega}(\mcal{C}^*f)
\right\rangle_{\rm HS}.
\]

On the commutative algebra $L^\infty(\bbR^{2D})$, the two functionals $\varphi_{+,0}$ and $\varphi_{-,0}$ differ by the constant factor $e^\omega$, so the corresponding relative modular operator is
multiplication by $e^\omega$. Hence
\begin{align*}
S_\ell^f(\varphi_{+,0},\varphi_{-,0})
&=
e^{-\omega/2}\ell(e^\omega)
\int_{\bbR^{2D}}|f(\alpha)|^2
\Husimi_\rho(\alpha)\diff\alpha\\
&=
\kappa_\omega
\int_{\bbR^{2D}}|f(\alpha)|^2
\Husimi_\rho(\alpha)\diff\alpha.
\end{align*}
Combining this with~\eqref{eq:Petz-abstract} proves \eqref{eq:Petz-consequence} when $\rho$ is faithful.

To remove the faithfulness assumption, fix a faithful density matrix $\tau$ on $\mcal{H}$ and apply the above result to $\rho_\eps := (1-\eps)\rho + \eps\,\tau$ to obtain 
\begin{equation}
\label{eq:Petz-epsilon}
\left\langle \mcal C^*f,\,
\Theta_{\rho_\varepsilon,\omega}(\mcal C^*f)
\right\rangle_{\rm HS}
\leq
\kappa_\omega
\int_{\bbR^{2D}}|f(\alpha)|^2
\Husimi_{\rho_\varepsilon}(\alpha)\diff\alpha.
\end{equation}
A standard limiting argument using continuity in trace norm on the left and dominated convergence on the right recovers the result for $\rho$.
\end{proof}

\section{Finite-action tangent directions and singular local geometry}
\label{sec:finite-action-tangents}

The finite-dimensional Carlen--Maas construction is Riemannian on the manifold of faithful states~\cite{carlen2020non}. Its local geometry can nevertheless become singular at nonfaithful states and, in infinite dimension, the inverse quadratic form defining the tangent action can be finite only on a proper subspace of trace-zero perturbations. The dual action in Eq.~\eqref{eq:quantum-dual-action} provides a precise way to describe this phenomenon. We define the space of finite-action tangent directions at $\rho$ by
\begin{equation}\label{eq:finite-action-tangent-space}
\operatorname{Tan}_{\rho,\beta}^{\rm fin} := \left\{\xi \in \mcal M_0 : \msf{a}_{\rho,\beta}(\xi) < \infty\right\}.
\end{equation}

In finite dimension, this space can be described explicitly. Let $\mcal K_{\rho,\beta}$ denote the positive semidefinite operator associated with the polarization of $Q_{\rho,\beta}$, so that
\begin{equation}
Q_{\rho,\beta}(A) = \langle A,\mcal K_{\rho,\beta}A\rangle_{\rm HS}
\end{equation}
for self-adjoint observables $A$.

\begin{lemma}
\label{lem:finite-action-tangent-space}
Suppose that $\mcal H$ is finite-dimensional. Then
\begin{equation}\label{eq:dual-action-pseudoinverse}
\msf{a}_{\rho,\beta}(\xi) = \begin{cases}
\langle \xi,\mcal K_{\rho,\beta}^{+}\xi\rangle_{\rm HS}, & \xi \in \operatorname{Ran}\mcal K_{\rho,\beta},\\[4pt] +\infty, & \xi \notin \operatorname{Ran}\mcal K_{\rho,\beta},
\end{cases}
\end{equation}
where $\mcal K_{\rho,\beta}^{+}$ is the Moore--Penrose pseudoinverse. In particular,
\begin{equation}
\operatorname{Tan}_{\rho,\beta}^{\rm fin} = \operatorname{Ran}\mcal K_{\rho,\beta}.
\end{equation}
\end{lemma}

\begin{proof}
Decompose the real Hilbert space of self-adjoint observables as
\[
\ker\mcal K_{\rho,\beta} \oplus (\ker\mcal K_{\rho,\beta})^\perp.
\]
If $\xi$ has nonzero pairing with an observable in $\ker\mcal K_{\rho,\beta}$, then scaling that observable in the variational formula~\eqref{eq:quantum-dual-action} shows that the supremum is infinite. Otherwise, $\xi \in (\ker\mcal K_{\rho,\beta})^\perp = \operatorname{Ran}\mcal K_{\rho,\beta}$, and completing the square gives Eq.~\eqref{eq:dual-action-pseudoinverse}.
\end{proof}

If $\rho$ is faithful and the gradient calculus is ergodic, in the sense that $Q_{\rho,\beta}(A) = 0$ only when $A$ is a multiple of the identity, then
\[
\ker\mcal K_{\rho,\beta} = \bbR\One.
\]
It follows that $\operatorname{Tan}_{\rho,\beta}^{\rm fin} = \mcal M_0$, and the metric is locally Riemannian. Thus the restriction to a proper space of finite-action directions is a boundary or infinite-dimensional phenomenon, rather than a property of every state.

To see what happens at a finite-rank state in the present oscillator model, let $P$ denote the support projection of $\rho$. The spectral formula for $\Theta_{\rho,\omega}$ gives
\begin{equation}\label{eq:Theta-support-compression}
\Theta_{\rho,\omega}(B) = P\,\Theta_{\rho,\omega}(PBP)P.
\end{equation}
Consequently, $Q_{\rho,\beta}(A)$ depends only on the finitely many compressed commutators
\begin{equation}
\mcal G_\rho(A) := \bigl(P[V_\ell,A]P\bigr)_{\ell \in \mcal I}.
\end{equation}
If $P$ has finite rank and its range lies in the common domains of the $V_\ell$ and $V_\ell^\dagger$, then $\mcal G_\rho$ has finite-dimensional range. Finiteness of the dual action requires $\Tr(A\xi) = 0$ whenever $\mcal G_\rho(A) = 0$, so $\operatorname{Tan}_{\rho,\beta}^{\rm fin}$ is contained in a finite-dimensional subspace of $\mcal M_0$.

For example, for one oscillator in the vacuum state $\rho = \ket 0\bra 0$, the compressed commutators depend only on the matrix elements $\langle 1|A|0\rangle$ and $\langle 0|A|1\rangle$. Hence
\begin{equation}
\operatorname{Tan}_{\ket 0\bra 0,\beta}^{\rm fin} = \operatorname{span}_{\bbR}\left\{\ket 1\bra 0 + \ket 0\bra 1,\, \rmi\ket 1\bra 0 - \rmi\ket 0\bra 1\right\}.
\end{equation}
A general unitary need not preserve these finite-action tangent spaces. For instance, a unitary that fixes $\ket 0$ and maps $\ket 1$ to $\ket 2$ maps the first direction above to
\[
|2\rangle \langle 0| + |0\rangle \langle 2|\,,
\]
which is not a finite-action direction at $\ket 0\bra 0$. By contrast, metaplectic unitaries preserve the canonical gradient module. In particular, Eq.~\eqref{eq:action-symplectic} implies that, at $\beta = 0$,
\begin{equation}
\mcal U_S\!\!\left[\operatorname{Tan}_{\rho,0}^{\rm fin}\right] = \operatorname{Tan}_{\mcal U_S[\rho],0}^{\rm fin}.
\end{equation}

When $\rho$ is nonfaithful, a formal affine perturbation $\rho + \varepsilon\xi$ need not remain positive. Local statements should therefore be formulated using admissible curves
\[
\rho_\varepsilon = \rho + \varepsilon\,\xi + O(\varepsilon^2)
\]
inside $\mcal D$. Moreover, the condition $\msf{a}_{\rho,\beta}(\xi) = +\infty$ does not by itself imply a universal power law for $d_{\mathrm{CM},2}(\rho,\rho_\varepsilon)$. For comparison, in the classical two-point logarithmic-mean metric, writing
\[
\rho_r = (1-r,r), \qquad \Lambda(1-r,r) = \frac{1-2r}{\log(1-r)-\log r}\,,
\]
one has
\begin{equation}
d(\rho_0,\rho_\varepsilon) \asymp \int_0^\varepsilon \frac{\diff r}{\sqrt{\Lambda(1-r,r)}} \asymp \varepsilon\sqrt{\log(1/\varepsilon)}.
\end{equation}
Thus the ratio of the distance to $\varepsilon$ diverges, but the singularity is logarithmic rather than a fixed fractional power.

This singular local geometry explains why an unrestricted infinitesimal Lipschitz quotient can probe properties of the transport metric rather than dynamical instability. The proof of Theorem~\ref{thm:precise-main} below does not require a particular local asymptotic for $d_{\mathrm{CM},2}$. It instead uses the uniform $O(\hbar^{1/2})$ semiclassical comparison and therefore only requires the initial separation to satisfy $\frac{d_p(\rho,\sigma)}{\hbar^{1/2}} \to \infty$.

\section{Classical optimal transport and expansion exponent}
\label{sec:classical-ot-exponent}

The definition of the classical expansion exponent $\lambda_{\rm exp}$ can also be phrased in terms of optimal transport.  The idea is to embed points as delta-function measures, using
\[
W_p(\delta_x, \delta_y) = |x-y|.
\]
The expansion coefficient can be redefined to superficially use optimal transport as follows:
\[
\lambda_{\rm exp} := \limsup_{T\to\infty} \sup_{x\not=y} \frac1T
\log\Big(\frac{W_p((\Phi_T)_{\#} \delta_x, (\Phi_T)_{\#}\delta_y)}{W_p(\delta_x,\delta_y)}\Big)
\]
where $(\Phi_T)_{\#}$ is the pushforward map on measures. Since $(\Phi_T)_{\#}\delta_x = \delta_{\Phi_T(x)}$, the supremum above is precisely the Lipschitz constant of the classical flow $\Phi_T$:
\[
\lambda_{\rm exp} := \limsup_{T\to\infty} \frac1T \log(\Lip(\Phi_T)).
\]

The natural extension of this idea would be to take the supremum not only over points $x\not=y$, but over general measures.  Doing so, we are led to define
\[
\tilde{\lambda}_{\rm exp} := \limsup_{T\to\infty} \sup_{\mu\not=\nu} \frac1T \log\Big(\frac{W_p((\Phi_T)_{\#} \mu, (\Phi_T)_{\#}\nu)}{W_p(\mu,\nu)}\Big),
\]
where the supremum is over probability measures $\mu$ and $\nu$. It is clear that $\tilde{\lambda}_{\rm exp} \geq \lambda_{\rm exp}$, simply because the supremum is now over a larger space.  The reverse inequality also holds, since
\begin{align*}
W_p((\Phi_T)_{\#}\mu,(\Phi_T)_{\#}\nu) \leq \Lip(\Phi_T)\, W_p(\mu,\nu).
\end{align*}
Therefore, $\tilde{\lambda}_{\rm exp} = \lambda_{\rm exp}$.

One might therefore think that one can simply replace the classical transport distance with a quantum distance and recover a suitable definition.  There is a subtle issue however, which necessitates the unusual order of limits in our definition.

This issue can best be illustrated for Hamiltonians that are compactly supported perturbations of a harmonic oscillator, with the quantum distance taken at $p=1$. Namely, let $H(x,p) = p^2 + x^2 + W(x)$, with $W\in C_c^\infty(\mathbb R^D)$, and
\[
d_{\mathrm{CM},1}(\rho,\sigma) := \sup_{\|A\|_{\rm Lip} \leq 1} \Tr[A(\rho-\sigma)] =: \|\rho-\sigma\|_{\rm KR}.
\]
The nice feature is that the $1$-Wasserstein distance can be interpreted in terms of a norm, the Kantorovich-Rubinstein norm. Suppose that the classically unstable dynamics is confined below some fixed energy scale $E$. Because the quantum Hamiltonian is confining, its spectral subspace below $E$ is finite dimensional at fixed $\hbar$. On this finite-dimensional subspace, the trace norm is equivalent to the Kantorovich--Rubinstein norm, so there exist constants $c,C>0$ such that
\[
c\,d_{\mathrm{CM},1}(\rho,\sigma) \leq \|\rho-\sigma\|_{\Tr} \leq C\,d_{\mathrm{CM},1}(\rho,\sigma).
\]

Therefore, letting $\rho(T)$ and $\sigma(T)$ be the evolution of $\rho$ and $\sigma$ by the
Schr\"{o}dinger evolution generated by $\hat{H}$, we have
\[
d_{\mathrm{CM},1}(\rho(T),\sigma(T))
\leq c^{-1} \|\rho(T)-\sigma(T)\|_{\Tr}
= c^{-1} \|\rho-\sigma\|_{\Tr} \leq c^{-1}C \,d_{\mathrm{CM},1}(\rho,\sigma).
\]
That is, there is a $T$-independent bound for the expansion ratio $d_{\mathrm{CM},1}(\rho(T),\sigma(T)) /d_{\mathrm{CM},1}(\rho,\sigma)$ among states $\rho$ and $\sigma$ in the bounded-energy sector.  On the other hand, the high energy dynamics is governed by the harmonic oscillatory $H(x,p) = x^2+p^2$ (up to a perturbation), so there is no reason to expect that one would see any exponential expansion in that sector.  Thus, it is reasonable to conjecture that the naive definition of the quantum Lyapunov exponent at finite
$\hbar$ vanishes, namely
\[
\lim_{T\to\infty} \sup_{\rho\not= \sigma} \frac1T \log\Big(\frac{d_{\mathrm{CM},1}(\rho(T),\sigma(T))}{d_{\mathrm{CM},1}(\rho,\sigma)}\Big) = 0.
\]
Although we sketched this argument in the case $p=1$ with one specific metric, it does seem that the general mechanism of compactness which comes from the finite-dimensionality of the relevant dynamical space is a generic obstacle to developing a fixed-$\hbar$ theory of quantum chaos.  It seems therefore that a general notion of quantum Lyapunov exponent requires considering some sequence of dynamical systems whose effective dimension tends to infinity, whether that be a semiclassical limit or a large $N$-type limit.

\section{Proof of the main theorem}

Now we state and prove the precise version of our main result Theorem~\ref{thm:lambda-eq}.
\begin{theorem}[Precise form of the main result]
\label{thm:precise-main}
Let $H\in S_2(1)$ be an admissible Hamiltonian as in Definition~\ref{def:admissible-H}. Then
\begin{equation}
\lambda_{\rm q} = \lambda_{\rm exp}\,,
\end{equation}
where $\lambda_{\rm exp}$ is the classical expansion coefficient defined in the main text. Let $\{\rho_\hbar\}_{\hbar>0}$ be a family of quantum states with classical phase-space limit $\mu$, in the sense of Eq.~\eqref{eq:rho-sc-lim}, and define
\begin{equation}
C_{jk}^{\rm cl}(T) := \int_{\mathbb R^{2D}}\left|\{\Phi_T(\alpha)_j,r_k\}_{\rm PB}\right|^2\diff\mu(\alpha).
\end{equation}
Then
\begin{equation}
\lambda_{\rm OTOC} = \limsup_{T\to\infty}\frac{1}{2T}\log\!\left[\sigma_{\max}\!\left(\mathbf C^{\rm cl}(T)\right)\right] \leq \lambda_{\rm exp}.
\end{equation}
\end{theorem}

The proof of this result relies on an estimate of the classical transport distance of Husimi distributions evolved by the Hamiltonian flow proved in~\cite{cotler2025egorov}.
\begin{theorem}[Consequence of Theorem 1.1 of~\cite{cotler2025egorov}]
\label{thm:wass-egorov}
Let $H\in S_2(1)$ be an admissible Hamiltonian, and let $\rho \in \mcal{D}_p$ be a density matrix and $\rho(t) := e^{-it\hat{H}/\hbar}\rho e^{it\hat{H}/\hbar}$, and let $\Phi_t$ be the classical flow by $H$.  Then
\[
W_p\!\left(\Husimi_{\rho(t)},(\Phi_t)_\#\Husimi_\rho\right)
\leq e^{Ct}\hbar^{1/2}\,.
\]
\end{theorem}

We are now ready to prove Theorem~\ref{thm:precise-main}.

\begin{proof}[Proof of Theorem~\ref{thm:precise-main} using Theorem~\ref{thm:wass-egorov}]
We first prove $\lambda_{\rm q} \geq \lambda_{\rm exp}$. Define
\begin{equation}
J(T) := \sup_{\alpha\neq\beta}\frac{|\Phi_T(\alpha)-\Phi_T(\beta)|}{|\alpha-\beta|}.
\end{equation}
For each $T$, choose $\alpha_T\neq\beta_T$ such that
\begin{equation}
\frac{|\Phi_T(\alpha_T)-\Phi_T(\beta_T)|}{|\alpha_T-\beta_T|} \geq \frac{1}{2}J(T).
\end{equation}
Let $\rho_{\hbar;T} = |\alpha_T\rangle\langle\alpha_T|$ and $\sigma_{\hbar;T} = |\beta_T\rangle\langle\beta_T|$. Since these coherent states are related by a phase-space translation, the translation-distance axiom gives
\begin{equation}
d_p(\rho_{\hbar;T},\sigma_{\hbar;T}) = |\alpha_T-\beta_T|.
\end{equation}
For fixed $T$, Theorem~\ref{thm:wass-egorov} together with Theorem~\ref{thm:upperbd1} implies
\begin{equation}
d_p\!\left(\rho_{\hbar;T}(T),|\Phi_T(\alpha_T)\rangle\langle\Phi_T(\alpha_T)|\right) \leq C_T\hbar^{1/2},
\end{equation}
and the same estimate holds with $\alpha_T$ replaced by $\beta_T$. Hence the triangle inequality gives
\begin{equation}
d_p(\rho_{\hbar;T}(T),\sigma_{\hbar;T}(T)) \geq |\Phi_T(\alpha_T)-\Phi_T(\beta_T)| - C_T\hbar^{1/2}.
\end{equation}
Since $|\alpha_T-\beta_T|$ is independent of $\hbar$, this family satisfies $d_p(\rho_{\hbar;T},\sigma_{\hbar;T}) = \omega(\hbar^{1/2})$, and therefore
\begin{equation}
\limsup_{\hbar\to0}\frac{d_p(\rho_{\hbar;T}(T),\sigma_{\hbar;T}(T))}{d_p(\rho_{\hbar;T},\sigma_{\hbar;T})} \geq \frac{1}{2}J(T).
\end{equation}
Taking the logarithm, dividing by $T$, and then sending $T\to\infty$ yields $\lambda_{\rm q}\geq\lambda_{\rm exp}$.

For the reverse inequality, Theorem~\ref{thm:upperbd1}, Theorem~\ref{thm:wass-egorov}, and the classical transport bound
\begin{equation}
W_p\bigl((\Phi_T)_\#\mu,(\Phi_T)_\#\nu\bigr) \leq J(T)W_p(\mu,\nu)
\end{equation}
give, for fixed $T$,
\begin{align}
d_p(\rho(T),\sigma(T))
&\leq W_p(\Husimi_{\rho(T)},\Husimi_{\sigma(T)}) + C\hbar^{1/2} \nonumber\\
&\leq W_p\bigl((\Phi_T)_\#\Husimi_\rho,(\Phi_T)_\#\Husimi_\sigma\bigr) + C_T\hbar^{1/2} \nonumber\\
&\leq J(T)W_p(\Husimi_\rho,\Husimi_\sigma) + C_T\hbar^{1/2} \nonumber\\
&\leq J(T)d_p(\rho,\sigma) + C_T\hbar^{1/2}.
\end{align}
For any $\hbar$-dependent family satisfying $d_p(\rho_\hbar,\sigma_\hbar) = \omega(\hbar^{1/2})$, it follows that
\begin{equation}
\limsup_{\hbar\to0}\frac{1}{T}\log\!\left(\frac{d_p(\rho_\hbar(T),\sigma_\hbar(T))}{d_p(\rho_\hbar,\sigma_\hbar)}\right) \leq \frac{1}{T}\log J(T).
\end{equation}
Taking the supremum over admissible families and then $T\to\infty$ gives $\lambda_{\rm q}\leq\lambda_{\rm exp}$.

Finally, define
\begin{equation}
C_{jk}^{(\hbar)}(T) := \frac{1}{\hbar^2}\Tr\!\left(\rho_\hbar\left|[\hat r_j(T),\hat r_k(0)]\right|^2\right).
\end{equation}
For fixed $T$, Egorov's theorem and the symbolic calculus give
\begin{equation}
\lim_{\hbar\to0}C_{jk}^{(\hbar)}(T) = C_{jk}^{\rm cl}(T).
\end{equation}
Since $J(T)$ is the Lipschitz constant of $\Phi_T$, the entries of the tangent map are bounded in magnitude by $J(T)$, and hence
\begin{equation}
\sigma_{\max}\!\left(\mathbf C^{\rm cl}(T)\right) \leq 2D\,J(T)^2.
\end{equation}
Therefore
\begin{equation}
\lambda_{\rm OTOC} \leq \lim_{T\to\infty}\frac{1}{2T}\log\!\left(2D\,J(T)^2\right) = \lambda_{\rm exp}.
\end{equation}
\end{proof}

\subsection{Commentary on $\lambda_{\rm OTOC}$}

The preceding result highlights the distinction between $\lambda_{\rm OTOC}$ and the global classical expansion coefficient $\lambda_{\rm exp}$. The former averages squared linearized growth over the classical phase-space distribution selected by the state, whereas $\lambda_{\rm exp}$ optimizes the finite-time expansion over phase space before taking the long-time limit. To make this distinction explicit, define
\begin{equation}
\chi_T(\alpha) := \frac{1}{T}\log\!\left\|\mathrm D\Phi_T(\alpha)\right\|_{\rm F},
\end{equation}
where $\|\cdot\|_{\rm F}$ denotes the Frobenius norm. Since
\begin{equation}
\sum_{j,k}C_{jk}^{\rm cl}(T) = \int_{\mathbb R^{2D}}\left\|\mathrm D\Phi_T(\alpha)\right\|_{\rm F}^2\diff\mu(\alpha),
\end{equation}
and the entries of $\mathbf C^{\rm cl}(T)$ are nonnegative, we have
\begin{equation}
\frac{1}{(2D)^2}\int_{\mathbb R^{2D}}\left\|\mathrm D\Phi_T(\alpha)\right\|_{\rm F}^2\diff\mu(\alpha) \leq \sigma_{\max}\!\left(\mathbf C^{\rm cl}(T)\right) \leq \int_{\mathbb R^{2D}}\left\|\mathrm D\Phi_T(\alpha)\right\|_{\rm F}^2\diff\mu(\alpha).
\end{equation}
The dimension-dependent constants disappear after taking the logarithm and dividing by $T$, so
\begin{equation}
\lambda_{\rm OTOC} = \limsup_{T\to\infty}\frac{1}{2T}\log\!\left[\int_{\mathbb R^{2D}}e^{2T\chi_T(\alpha)}\diff\mu(\alpha)\right].
\end{equation}
Thus $\lambda_{\rm OTOC}$ is a state-weighted second-moment exponent, consistent with earlier observations that semiclassical OTOC growth probes moments of finite-time trajectory divergence rather than the usual trajectory-averaged Lyapunov exponent~\cite{Rozenbaum:2016mmv, PradoReynoso:2022jfh}. By contrast, up to the immaterial choice of matrix norm,
\begin{equation}
\lambda_{\rm exp} = \lim_{T\to\infty}\sup_{\alpha\in\mathbb R^{2D}}\chi_T(\alpha).
\end{equation}
This immediately gives $\lambda_{\rm OTOC}\leq\lambda_{\rm exp}$ and explains why the inequality can be strict. In settings where $\lambda_{\rm exp}=\lambda_{\rm cl}$, the same bound also gives $\lambda_{\rm OTOC}\leq\lambda_{\rm cl}$.

The saturation condition can be expressed in terms of the weight assigned to nearly maximally expanding regions. For $\varepsilon>0$, define
\begin{equation}
E_T(\varepsilon) := \left\{\alpha\in\operatorname{supp}\mu:\chi_T(\alpha)\geq\lambda_{\rm exp}-\varepsilon\right\}.
\end{equation}
Restricting the phase-space integral to this set gives
\begin{equation}
\lambda_{\rm OTOC}\geq\lambda_{\rm exp}-\varepsilon+\frac{1}{2}\liminf_{T\to\infty}\frac{1}{T}\log\mu\bigl(E_T(\varepsilon)\bigr).
\end{equation}
It follows that $\lambda_{\rm OTOC}=\lambda_{\rm exp}$ whenever
\begin{equation}
\lim_{T\to\infty}\frac{1}{T}\log\mu\bigl(E_T(\varepsilon)\bigr)=0
\end{equation}
for every $\varepsilon>0$. Thus nearly maximally expanding trajectories need not have positive measure; it is enough that their weight not be exponentially small in time.

Suppose now that $\mu$ is invariant and ergodic under the classical flow. Oseledets' theorem gives a top Lyapunov exponent $\lambda_{\rm top}(\mu)$ that is constant for $\mu$-almost every initial point. Ergodicity alone does not imply $\lambda_{\rm OTOC}=\lambda_{\rm exp}$. The typical exponent $\lambda_{\rm top}(\mu)$ may be strictly smaller than the global expansion rate, and rare finite-time fluctuations can affect the exponential second moment even when the asymptotic exponent is constant almost everywhere. A sufficient, but stronger, condition is uniform convergence of the finite-time expansion rate on the support of $\mu$,
\begin{equation}
\lim_{T\to\infty}\sup_{\alpha\in\operatorname{supp}\mu}\left|\chi_T(\alpha)-\lambda_{\rm top}(\mu)\right| = 0\,.
\end{equation}
Under this condition, $\lambda_{\rm OTOC}=\lambda_{\rm top}(\mu)$, and the bound is saturated if $\lambda_{\rm top}(\mu)=\lambda_{\rm exp}$. For a nonergodic measure, the averaging before the logarithm similarly means that the OTOC is dominated by the components or rare phase-space regions that achieve the best balance between rapid growth and statistical weight.

A simple illustration is provided by a particle in a one-dimensional double-well potential with Hamiltonian $H(x,p) = \frac{p^2}{2} + V(x)$, where $V$ has a local maximum at $x=x_*$. Writing
\begin{equation}
V(x) = V(x_*)-\frac{\omega_*^2}{2}(x-x_*)^2+O\!\left((x-x_*)^3\right), \quad \omega_* := \sqrt{-V''(x_*)}\,,
\end{equation}
the point $(x_*,0)$ is a hyperbolic fixed point with local instability exponent $\omega_*$. Generic bounded trajectories in this one-degree-of-freedom system are regular and have vanishing asymptotic Lyapunov exponent, while the hyperbolic fixed point itself has positive exponent $\omega_*$. The global expansion coefficient, and hence also the usual maximal Lyapunov exponent when no larger instability is present elsewhere, is therefore positive even though the instability is not representative of typical trajectories.

The behavior of the OTOC depends on the phase-space distribution. If $\mu$ is supported a positive distance away from the separatrix, the OTOC does not sample the saddle instability and its asymptotic exponent can vanish. On the other hand, suppose that $\mu$ has a smooth, nonvanishing density near the saddle. Initial conditions that remain in the linear saddle region for a time $T$ occupy an exponentially narrow strip whose measure scales schematically as $e^{-\omega_*T}$, while their squared sensitivity grows as $e^{2\omega_*T}$. Their contribution to the averaged squared sensitivity therefore scales as
\begin{equation}
\sim e^{-\omega_*T}e^{2\omega_*T} = e^{\omega_*T}.
\end{equation}
In the normalization used here, this local estimate corresponds to a positive OTOC exponent $\omega_*/2$, even though almost every trajectory is nonchaotic. If the state is instead concentrated near the saddle so that the unstable region is not exponentially suppressed, the OTOC exponent can approach the full saddle exponent $\omega_*$. This saddle-dominated exponential growth of OTOCs in classically integrable systems has been discussed extensively in the literature~\cite{xu2020scrambling,hashimoto2020exponential,trunin2023quantum}.

This saddle-dominated OTOC growth is fully consistent with the quantum-to-classical correspondence. The quantum OTOC converges to a classical phase-space average of squared tangent-map elements; the distinction arises because this average is taken before the logarithm. The quantum optimal transport exponent and the OTOC exponent therefore answer different questions: $\lambda_{\rm q}$ optimizes over initial state perturbations and recovers the global classical expansion rate $\lambda_{\rm exp}$, whereas $\lambda_{\rm OTOC}$ measures a state-weighted second moment of the tangent dynamics. The two coincide when nearly maximally expanding trajectories contribute without exponential suppression. In settings where $\lambda_{\rm exp}=\lambda_{\rm cl}$, this also gives $\lambda_{\rm q} = \lambda_{\rm cl} = \lambda_{\rm OTOC}$ under the same saturation condition.

\end{document}